\documentclass{amsproc}

\usepackage[T2A]{fontenc}
\usepackage[cp1251]{inputenc}
\usepackage[english]{babel}

\usepackage{amsfonts}
\usepackage{amsthm}
\usepackage{amssymb}
\usepackage{amsmath}
\usepackage{amsopn}
\usepackage{bm}
\usepackage{array}
\usepackage{color}
\usepackage{mathdots}
\usepackage{graphicx}

\newtheorem{prop}{Proposition}
\newtheorem{theo}{Theorem}

\theoremstyle{definition}

\newtheorem{rem}{Remark}

\usepackage{cite}
\allowdisplaybreaks[4]

\DeclareMathOperator{\ad}{ad}
\DeclareMathOperator{\Ad}{Ad}
\DeclareMathOperator{\spanOp}{span}
\DeclareMathOperator{\tr}{tr}
\DeclareMathOperator{\dn}{dn}
\DeclareMathOperator{\cn}{cn}
\DeclareMathOperator{\sn}{sn}

\newcommand{\x}{\mathrm{x}}
\newcommand{\rmt}{\mathrm{t}}

\newcommand{\rmd}{\mathrm{d}}
\DeclareMathOperator{\modR}{mod}

\newcommand{\M}{\mathcal{M}}

\newcommand{\rmb}{\mathrm{b}}
\newcommand{\rma}{\mathrm{a}}

\newcommand{\B}{\textsf{B}}
\newcommand{\A}{\textsf{A}}
\newcommand{\Pp}{\textsf{P}}
\newcommand{\Qp}{\textsf{Q}}

\renewcommand{\H}{\textsf{H}}
\newcommand{\X}{\textsf{X}}
\newcommand{\Y}{\textsf{Y}}
\newcommand{\Z}{\textsf{Z}}

\DeclareMathOperator{\Complex}{\mathbb{C}}
\DeclareMathOperator{\Real}{\mathbb{R}}
\DeclareMathOperator{\Natural}{\mathbb{N}}
\DeclareMathOperator{\Integer}{\mathbb{Z}}

\DeclareMathOperator{\Jac}{\mathfrak{J}}

\DeclareMathOperator*{\res}{res}
\DeclareMathOperator{\const}{const}

\newcommand{\w}{\textsf{w}}

\newcommand{\mFr}{\mathfrak{m}}

\newcommand{\kFr}{\mathfrak{k}}
\newcommand{\I}{\mathcal{I}}
\newcommand{\J}{\mathcal{J}}

\DeclareMathOperator{\ReN}{\mathrm{Re}}
\DeclareMathOperator{\ImN}{\mathrm{Im}}

\title[Reality conditions for KdV and exact quasi-periodic solutions]{Reality conditions for the KdV equation and
exact quasi-periodic solutions in finite phase spaces}
\author{Julia Bernatska}
\address{University of Connecticut, Department of Mathematics}
\email{jbernatska@gmail.com}

\begin{document}
 
\maketitle

\begin{abstract}
In the present paper reality conditions for quasi-periodic solutions of the KdV equation
are determined completely. As a result, solutions in the form of non-linear waves
can be plotted and investigated. 

The full scope of obtaining finite-gap solutions of the KdV equation is presented. 
It is proven that the multiply periodic $\wp_{1,1}$-function on the Jacobian variety of a hyperelliptic curve
of arbitrary genus serves as the finite-gap solution, the genus coincides with the number of gaps.
The subspace of the Jacobian variety where $\wp_{1,1}$, as well as other $\wp$-functions, are 
bounded and real-valued is found in any genus.
This result covers every finite phase space of the KdV hierarchy,
and can be extended to other completely integrable  equations. 
A method of effective computation of this type of solutions is suggested, and illustrated in genera $2$ and $3$.
\end{abstract}

% 35B15  Almost and pseudo-almost periodic solutions to PDEs
% 35B30 Dependence of solutions to PDEs on initial and/or boundary data and/or on parameters of PDEs [See also 37Cxx
% 35Q53 KdV equations (Korteweg-de Vries equations) {For dynamical systems and ergodic theory, see 37K10}

% 37J35 Completely integrable finite-dimensional Hamiltonian systems, integration methods, integrability tests
% 37J37 Relations of finite-dimensional Hamiltonian and Lagrangian systems with Lie algebras and other algebraic structures
% 37J38 Relations of finite-dimensional Hamiltonian and Lagrangian systems with algebraic geometry, complex analysis, special functions
% 37K10 Completely integrable infinite-dimensional Hamiltonian and Lagrangian systems, integration methods, inte-
%    grability tests, integrable hierarchies (KdV, KP, Toda, etc.)
%  37K15 Inverse spectral and scattering methods for infinite-dimensional Hamiltonian and Lagrangian systems
% 37K20 Relations of infinite-dimensional Hamiltonian and Lagrangian dynamical systems with algebraic geometry,
%  complex analysis, and special functions [See also 14H70]

%======================================
\section{Introduction}

The Korteweg---de Vries equation (KdV) arose in the $19$th century in connection with the theory
of waves in shallow water\footnote{The KdV equation was mentioned in the footnote on page 360
in Boussinesq, J., Essai sur la theorie des eaux courantes, Memoires presentes par divers savants, 
l'Acad. des Sci. Inst. Nat. France, XXIII (1877), pp.\,1--680. However, D. J. Korteweg  and G. de Vries,  (1895)
gave the full explanation in \cite{KdV}.}, see \cite{KdV}. The equation also describes the
propagation of waves with weak dispersion in various nonlinear media, see \cite{KK1971}. 
The conventional form of KdV is
\begin{equation}\label{KdVEq}
w_{\rmt} = 6 w w_{\x} + w_{\x\x\x}.
\end{equation}
This equation is scale-invariant, that is by scaling $\rmt$, $\x$, and $w$
one can change the coefficients of the three terms  arbitrarily.

The first two solutions: the one-soliton solution, and the simplest non-linear wave solution --- were
found by Korteweg and de Vries  in \cite{KdV}.

In \cite{GGKM1967} a remarkable procedure of finding solutions of KdV,
known as the inverse scattering method, was discovered. Although 
an elegant form of solutions was suggested, the problem was merely transformed into
the Gel'fand---Levitan integral equation.
Soon, it was shown that 
the KdV equation admits the Lax representation \cite{Lax1968},
and possesses a sequence of integrals of motion \cite{MGK1968}, which tends to be infinite. 
In \cite{ZF1971} it was proven, that the sequence of integrals of motion is infinite, 
and so KdV was called a \emph{completely integrable} hierarchy of hamiltonian systems. 
The notion of a \emph{finite-gap solution} arose in \cite{Nov1974}; such a solution
 lives on a finite-dimesional phase space. And within the infinite hierarchy,
a solution in any  dimension can be constructed.
Higher KdV equations were also introduced in  \cite{Nov1974}.

In \cite{Hir1971} a method of constructing multi-soliton solutions, known as Hirota's method,
was suggested. Soon after that,  these multi-soliton solutions were obtained by the inverse scattering method
in the case of no reflection of incoming waves \cite{GGKM1974}.

An alternative  method for solving the KdV equation, and a wide variety of other integrable equations,
comes from the theory of semi-simple Lie groups. 
An integrable equation arises within a hierarchy of hamiltonian systems 
on orbits of coadjoint representation of a loop  group,  
see in \cite{Adler78, MF1978,  Kos1979, AdlMoer80, Fom1981}
how the method was developed.
This method is known as the \emph{orbit approach}. It leads to algebraic integration,
and produces solutions in terms of functions which uniformize the spectral curve of
a hamiltonian system in question.

A finite-gap solution of the KdV equation was suggested in \cite{ItsMat1975}:
\begin{equation}\label{KdVSolLogTheta}
w(\x,\rmt) = 2 \partial_\x^2 \log \theta(\bm{U} \x + \bm{W} \rmt + \bm{D}) + 2c,
\end{equation}
in connection with the spectral problem for the Schr\"{o}dinger operator. 
The constant $c$ and constant vectors $\bm{U}$, $\bm{W}$,  $\bm{D}$ were not specified therein. In \cite{Dub1981}
a procedure of reconstructing $\bm{U}$ and $\bm{W}$ was suggested, while the constant $c$ was omitted.
A solution in the form \eqref{KdVSolLogTheta} is defined more accurately  in  \cite[p.\,65--66]{bbeim1994}.

Another form of an exact  finite-gap solution, see \cite{belHKF}, 
 is given by a multiply periodic $\wp_{1,1}$-function, defined on the Jacobian variety of a hyperelliptic curve of any genus, 
 the genus coincides with the number of gaps. 
 The latter solution sheds light on constant quantities in \eqref{KdVSolLogTheta},
 see Remark~\ref{r:Comparison} for more details.
  
 A solution in the form of $\wp_{1,1}$-function is complex-valued and satisfies
 \eqref{KdVEq} with arbitrary complex $\rmt$ and $\x$.  
Until now, real conditions remained an open problem for this type of solutions. 
In the present paper the problem of real conditions for $\wp_{1,1}$-solution is solved completely.
Simultaneously, the locus in the Jacobian variety where $\wp$-functions are bounded and real-valued is found
in the case of a curve with all real branch points.

In the recent decades, efficient computation and
graphical representation of the KdV solutions have been addressed. 
In \cite{FK2006} solutions of the KdV and KP equations on hyperelliptic curves of genera $2$, $4$, $6$
are computed in Matlab.  Solutions of the form \eqref{KdVSolLogTheta}
are used, and periods are calculated by expanding the integrands as a series of Chebyshev polynomials, 
and then integrating the polynomials in an appropriate way. This method of
 spectral approximation was introduced in \cite{FK2004}, as an alternative to the known 
 numerical tools of studying theta-functional solutions. The most popular among the latter is 
 the package \texttt{algcurves}  in Maple  \cite{DP2011}, which calculates 
characteristic quantities of Riemann surfaces, such as homology basis, not normalized period
 matrices, and the Riemann period matrix. The most recent development in this direction is 
 a new package of computing theta functions and  derivatives in Julia \cite{AC2021}. 
 Therein a review of the existing computational tools is presented.
 
Solutions of the KdV and mKdV equations in terms of $\wp_{1,1}$-function 
are computed and graphically represented in  \cite{Mats2023,MatsDNA2022}.
Function $\wp_{1,1}$ is expressed in terms of a divisor on a hyperelliptic curve, 
which gives a solution of the Jacobi inversion problem. 
Solving the equation is reduced to computing the Abel image of this divisor
by means of Euler's  numerical quadrature. A path in the 
Jacobian variety is constructed numerically to satisfy the reality condition for the solution.
This method avoids computing period matrices.

In the present paper, an analytical method of computing $\wp$-functions is used,
see \cite{BerCalc24} for more details. 
The method is based on analytical construction of the Riemann surface of a curve.
It produces period matrices of the first and second kinds, required for computation
of $\wp$-functions, and allows to compute the Abel image of any point of the curve.
The method was developed to supply
 $\wp$-functions with appropriate and convenient computational tools.
Wolfram Mathematica, designed for symbolic computation, is used for calculation
and graphical representation.

The  paper presents the full scope of obtaining  finite-gap solutions of the KdV equation.
We start with constructing the hierarchy of hamiltonian systems of the KdV equation, 
then explain algebraic integration in detail,
and prove that the $\wp_{1,1}$-function serves as a finite-gap solution in any genus. 
From the hierarchy we find the accurate relation between the arguments of
$\wp_{1,1}$ and variables $\rmt$ and $\x$. Then we find the domain in the Jacobian variety where
$\wp_{1,1}$ is bounded and real-valued. This result is proven in arbitrary genus.
Finally, quasi-periodic solutions of the KdV equation in genera $2$ and $3$ are illustrated by plots.

The paper is organized as follows. In Preliminaries we recall all notions related to uniformization
of hyperelliptic curves: the standard not normalized differentials of the first kind 
and associated to them differentials of the second kind,
Abel's map, the theta and sigma functions, the Jacobi inversion problem, 
and also briefly explain the construction of hamiltonian systems on coadjoint orbits of a loop group.
In section~\ref{s:KdVHier} hamiltonian systems which form the KdV hierarchy are developed.
Section~\ref{s:SoV} is devoted to separation of variables. In section~\ref{s:AGI} 
algebro-geometric integration is explained in application to the KdV hierarchy.
Finally, section~\ref{s:RealCond} presents new results on finding bounded and
real-valued quasi-periodic solutions of the KdV equation in any finite phase space.
In section~\ref{s:NLW}  a method of effective computation of $\wp$-functions is presented, 
and quasi-periodic solutions
are illustrated in genera $1$, $2$, and $3$.

%======================================
\section{Preliminaries}
%-------------------------------------
\subsection{Hyperelliptic curves}\label{ss:HC}
Let a non-degenerate hyperelliptic curve $\mathcal{V}$ of genus $g$ 
be defined\footnote{A $(2,2g+1)$-curve serves as a 
canonical form of hyperelliptic curves of genus $g$.} by
\begin{equation}\label{V22g1Eq}
f(x,y) \equiv -y^2 + x^{2 g+1} + \sum_{i=1}^{2g} \lambda_{2i+2} x^{2g-i} = 0.
\end{equation}

Let $(e_i,0)$, $i=1$, \ldots, $2g+1$, be finite branch points of the curve \eqref{V22g1Eq}.
In what follows, we denote branch points simply by $e_i$.
Homology basis is defined after H.\,Baker \cite[\S\,200]{bakerAF}.
One can imagine a continuous path through all branch points, which starts at infinity and ends at infinity,
see the orange line on fig.~\ref{cyclesOdd}. 
\begin{figure}[h]
\includegraphics[width=0.6\textwidth]{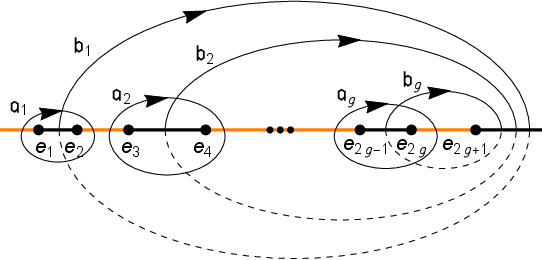} 
\caption{Cuts and cycles on a hyperelliptic curve.} \label{cyclesOdd}
\end{figure}
The branch points are enumerated along the path.
Fig.~\ref{cyclesOdd} represents the case of all real branch points, on which  the present paper focuses.
Cuts are made between points $e_{2k-1}$ and $e_{2k}$ with $k$ from $1$ to $g$. 
One more cut starts at  $e_{2g+1}$ and ends at infinity.
Canonical homology cycles are defined as follows.
Each $\mathfrak{a}_k$-cycle encircles the cut $(e_{2k-1},e_{2k})$, $k=1$, \ldots $g$,
and each $\mathfrak{b}_k$-cycle emerges from the cut $(e_{2k-1},e_{2k})$ 
and enters the cut $(e_{2g+1},\infty)$, see  fig.~\ref{cyclesOdd}.

Let $\rmd u = (\rmd u_1,\rmd u_3, \dots, \rmd u_{2g-1})^t$ be not normalized differentials
of the first kind, and $\rmd r = (\rmd r_1,\rmd r_3, \dots, \rmd r_{2g-1})^t$ be
differentials of the second kind associated with the first kind differentials, 
see \cite[\S\,138]{bakerAF} for more detail.
Actually, 
\begin{subequations}\label{Diffs}
\begin{align}
& \rmd u_{2n-1} =  \frac{x^{g-n} \rmd z}{\partial_y f(x,y)},\quad n=1,\dots, g,\label{K1DifsGen} \\
& \rmd r_{2n-1} =  \frac{ \rmd x}{\partial_y f(x,y)} \bigg(x^{g+n-1} 
+ \sum_{j=1}^{2n-2} (2n-1-j) \lambda_{2j} x^{g+n-1-j} \bigg), \label{K2DifsGen}
\end{align}
\end{subequations}
where $\partial_y f(x,y) = -2 y$. The indices of $\rmd u_{2n-1}$ display the orders of zeros, 
and the indices of $\rmd r_{2n-1}$ display the orders of poles.

Let  not nomalized periods along the canonical cycles $\mathfrak{a}_k$, $\mathfrak{b}_k$, $k=1$, \ldots, $g$, 
be defined as follows
\begin{gather}\label{NNormPerMatr1}
  \omega_k = \oint_{\mathfrak{a}_k} \rmd u,\qquad\qquad
  \omega'_k = \oint_{\mathfrak{b}_k} \rmd u.
\end{gather}
The vectors $\omega_k$, $\omega'_k$ 
form  first kind period matrices $\omega$, $\omega'$, respectively.

The corresponding normalized period matrices are $1_g$, $\tau$,
where $1_g$ denotes the identity matrix of size $g$, 
and $\tau = \omega^{-1}\omega'$.  Matrix $\tau$ is symmetric with a positive imaginary part: 
$\tau^t=\tau$, $\ImN \tau >0$,
that is $\tau$ belongs to the Siegel upper half-space. 
The normalized holomorphic differentials are denoted by
\begin{gather*}
 \rmd v = \omega^{-1} \rmd u.
\end{gather*}

%-------------------------------------
\subsection{Abel's map}
The vectors $\omega_k$, $\omega'_k$  serve as generators of a period lattice $\mathfrak{P}$.
Then $\Jac = \Complex^g/\mathfrak{P}$ is the Jacobian variety of the curve \eqref{V22g1Eq}.
Let $u=(u_1,u_3,\dots,u_{2g-1})^t$ be a  point of $\Jac$.

Let  the Abel map be defined by
\begin{gather*}%\label{AbelM}
 \mathcal{A}(P) = \int_{\infty}^P \rmd u,\qquad P=(z,y)\in \mathcal{V}.
\end{gather*}
The Abel map of a positive divisor $D = \sum_{i =1}^n (x_i,y_i)$ is defined by
\begin{gather*}%\label{AbelMDiv}
 \mathcal{A}(D) = \sum_{i =1}^n \int_{\infty}^{(x_i,y_i)} \rmd u.
\end{gather*}
The map is one-to-one on the $g$-th symmetric power of the curve:
 $\mathcal{A}: \mathcal{V}^g \mapsto \Jac$.

%----------------------------------------
\subsection{Theta function}
The Riemann theta function  is defined by 
\begin{gather}\label{ThetaDef}
 \theta(v;\tau) = \sum_{n\in \Integer^g} \exp \big(\imath \pi n^t \tau n + 2\imath \pi n^t v\big).
\end{gather}
This function is supposed to be related to the curve \eqref{V22g1Eq}, %$\mathcal{C}$ 
it depends on the normalized coordinates
 $v = \omega^{-1}u$, $u\in \Jac$, and periods $\tau = \omega^{-1}\omega'$.
 Let
\begin{equation}\label{ThetaDefChar}
 \theta[\varepsilon](v;\tau) = \exp\big(\imath \pi (\varepsilon'{}^t/2) \tau (\varepsilon'/2)
 + 2\imath \pi  (v+\varepsilon/2)^t \varepsilon'/2\big)  \theta(v+\varepsilon/2 + \tau \varepsilon'/2;\tau)
\end{equation}
be the theta function with characteristic $[\varepsilon]= (\varepsilon', \varepsilon)^t$.
A characteristic $[\varepsilon]$ is a $2\times g$ matrix, all components of $\varepsilon$, and $\varepsilon'$
are real values within the interval $[0,2)$. Modulo ($\modR$) $2$ addition
is defined on characteristics.

Every point $u$ within a fundamental domain of the Jacobian variety $\Jac$ 
can be represented by its characteristic $[\varepsilon]$
defined as follows
\begin{equation*}
u =  \tfrac{1}{2}  \omega \varepsilon +  \tfrac{1}{2}  \omega' \varepsilon'.
\end{equation*}
Abel's images of branch points are described by 
characteristics with integer components, as well as Abel's image of any combination of branch points.
An integer characteristic $[\varepsilon]$ is odd whenever $\varepsilon^t \varepsilon'  = 0$ ($\modR 2$), 
and even whenever $\varepsilon^t \varepsilon' = 1$ ($\modR 2$). A theta function with characteristic
has the same parity as its characteristic.

%---------------------------------------------------------
\subsection{Sigma function}\label{ss:SigmaFunct}
The modular invariant entire function on $\Jac$ is called the sigma function.
In the present paper we define it by the relation with the  theta function:
\begin{equation}\label{SigmaThetaRel}
\sigma(u) = C \exp\big({-}\tfrac{1}{2} u^t \varkappa u\big) \theta[K](\omega^{-1} u;  \omega^{-1} \omega').
\end{equation}
Note, that the sigma function is defined in terms of not normalized coordinates $u$, and associated with
not normalized period matrices of the first kind $\omega$, $\omega'$, and  of the second kind $\eta$, $\eta'$.
The latter matrices are formed by the vectors 
\begin{gather*}
  \eta_k = \oint_{\mathfrak{a}_k} \rmd r ,\qquad\qquad
  \eta'_k = \oint_{\mathfrak{b}_k} \rmd r,
\end{gather*}
respectively. Then $\varkappa = \eta \omega^{-1}$ is a symmetric matrix.

In what follows we use multiply periodic $\wp$-functions 
\begin{gather*}%\label{WPdef}
\wp_{i,j}(u) = -\frac{\partial^2 \log \sigma(u) }{\partial u_i \partial u_j },\qquad
\wp_{i,j,k}(u) = -\frac{\partial^3 \log \sigma(u) }{\partial u_i \partial u_j \partial u_k}.
\end{gather*}

For constructing series representation of the sigma function see \cite{bl2008}.

%-------------------------------------
\subsection{Jacobi inversion problem}
A solution of the Jacobi inversion problem on a hyperelliptic curve is proposed in  \cite[Art.\;216]{bakerAF},
see also \cite[Theorem 2.2]{belHKF}.
Let $u = \mathcal{A}(D )$ be the Abel image of  a non-special positive divisor  
$D \in \mathcal{V}^g$. Then $D$ is uniquely defined by the system of equations 
\begin{subequations}\label{EnC22g1}
\begin{align}
&\mathcal{R}_{2g}(x,y;u) \equiv x^{g} -  \sum_{i=1}^{g} x^{g-i}  \wp_{1,2i-1}(u) = 0, \label{R2g}\\ 
&\mathcal{R}_{2g+1}(x,y;u) \equiv 2 y + \sum_{i=1}^{g} x^{g-i}  \wp_{1,1,2i-1}(u) = 0. \label{R2g1}
\end{align}
\end{subequations}
Here and below $\mathcal{R}_{n}$ denotes an entire rational function of weight $n$  on the curve.

%-------------------------------------
\subsection{Characteristics and partitions}
Let $\mathcal{S} = \{0,1,2,\dots, 2g+1\}$ be the set of indices of all branch points of a 
hyperellipitic curve of genus $g$, and $0$ stands for the branch point at infinity.
According to \cite[\S\,202]{bakerAF} 
all half-period characteristics are represented by 
partitions of $\mathcal{S}$ of the form $\I_\mFr\cup \J_\mFr$ with $\I_\mFr = \{i_1,\,\dots,\, i_{g+1-2\mFr}\}$
and $\J_\mFr = \{j_1,\,\dots,\, j_{g+1+2\mFr}\}$, where $\mFr$ runs from $0$ to $[(g+1)/2]$, 
and $[\cdot]$ denotes the integer part. 
Index $0$, corresponding to infinity, is usually omitted in sets, 
it is also omitted in  computation of cardinality of a set.

Denote by $[\varepsilon(\I)] = \sum_{i\in\I} [\varepsilon_i]$ $(\modR 2)$ the characteristic of
\begin{gather*}
 \mathcal{A} (\I) = \sum_{i\in\I} \mathcal{A}(e_i) 
 = \omega \Big(\tfrac{1}{2} \varepsilon(\I)  + \tfrac{1}{2} \tau \varepsilon'(\I) \Big).
\end{gather*}
Characteristics of $2g+1$ branch points serve as a basis for constructing
all $2^{2g}$ half-period characteristics.
Below a partition is  referred to by the part of less cardinality, denoted by $\I$.

Introduce also characteristic 
$$[\I] = [\varepsilon(\I)] + [K]$$ 
of $\mathcal{A} (\I) + K$,
where $K$ denotes the vector of Riemann constants.
Characteristic $[K]$ of the vector of Riemann constants equals 
the sum of all odd characteristics of branch points, there are $g$ such characteristics,
see \cite[\S\,200, 202]{bakerAF}.
In the basis of canonical cycles introduced by fig.~\ref{cyclesOdd} we have
\begin{gather*}
 [K] = \sum_{k=1}^g [\varepsilon_{2k}].
\end{gather*}

Let $\I_\mFr\cup \J_\mFr$ be a partition introduced above, then $[\I_\mFr]=[\J_\mFr]$.
Characteristics $[\I_\mFr]$ of even multiplicity $\mFr$ are even, and of odd $\mFr$ are odd.
According to the Riemann vanishing theorem, $\theta(v+\mathcal{A} (\I_\mFr)+K; \tau)$ 
vanishes to order $\mFr$ at $v\,{=}\,0$.
Number $\mFr$ is called \emph{multiplicity}.
Characteristics of multiplicity $0$ are called \emph{non-singular even characteristics}.
Characteristics of multiplicity $1$ are called \emph{non-singular odd}.
All other characteristics are called \emph{singular}.

%---------------------------------------------------------
%[according to \cite{Holod82}]
\subsection{Hamiltonian systems on coadjoint orbits}\label{ss:HSCO}
Let $\mathfrak{g}$ be a semi-simple Lie algebra, and 
$\mathfrak{g}^\ast$ be the dual algebra to $\mathfrak{g}$
with respect to a bilinear form $\langle \cdot , \cdot \rangle: \mathfrak{g}^\ast \times \mathfrak{g} \mapsto \Complex$.
That is, $\mathfrak{g}^\ast$  is the space of $1$-forms $\Psi$ on $\mathfrak{g}$.

An orbit of coadjoint action of the corresponding Lie group\footnote{Here $\exp$ 
denotes the exponential map from a Lie algebra to 
a Lie group.} $G = \exp(\mathfrak{g})$  is defined by
\begin{gather*}
\mathcal{O} = \{\Psi \mid \Psi = \Ad^\ast_g \Phi,\ g \in G\},
\end{gather*}
where  $\Phi$  denote an initial point in $\mathfrak{g}^\ast$.
If $S$ is the stationary subgroup of $\Phi$ in $G$, then
$$ \mathcal{O}  = G / S.$$
%In what follows, we denote the orbit by $\mathcal{O}$, omitting the index.

The tangent space at $\Psi \in \mathcal{O}$ is 
$$T_{\Psi} = \{\ad_A^\ast \Psi \mid A \in \mathfrak{g}\},$$
where $\ad_A^\ast \Psi  = - \ad_A \Psi  = [\Psi, A]$.
If $A \in \mathfrak{s}_{\Psi}$, where $\exp(\mathfrak{s}_{\Psi}) = S_{\Psi}$
is the stationary subgroup of $\Psi$, then $\ad_A^\ast \Psi = 0$. Thus,
$$T_{\Psi} = \mathfrak{g} /  \mathfrak{s}_{\Psi}.$$
When $\Psi$ runs the orbit $\mathcal{O}$, and $A\in \mathfrak{g} /  \mathfrak{s}_{\Psi}$ is fixed, 
tangent vectors $\ad_A^\ast \Psi$
draw a vector field $\widetilde{A}$ on the orbit.
Denote by $\rho$ the map $A \mapsto \widetilde{A}$. Let 
$$T(\mathcal{O}) = \{\widetilde{A} \overset{\rho}{=} 
\ad_A^\ast \Psi \mid A \in \mathfrak{g}/  \mathfrak{s}_{\Psi}, \Psi\in \mathcal{O}\} $$ 
be the tangent space of the orbit $\mathcal{O}$.

At a point $\Psi \in \mathfrak{g}^\ast$ a skew-symmetric $2$-form
$\varpi_{\Psi}: T(\mathcal{O}) \times T(\mathcal{O}) \mapsto \Complex$ is defined
by the rule
\begin{equation}\label{2form}
\varpi_{\Psi} (\widetilde{A},\widetilde{B}) = \langle \Psi, [A,B] \rangle.
\end{equation}
The form $\varpi_{\Psi}$  is non-degenerate and closed, 
that follows immediately from the definition \eqref{2form}.
The form $\varpi$ is $G$-invariant: that is, it does not change when $\Psi$
runs the orbit, since vector fields $\widetilde{A}$, $\widetilde{B}$ transform in accordance with
 $\Psi$, namely $\widetilde{A}  \overset{\rho}{=} \ad_A^\ast \Psi$, $\widetilde{B}  \overset{\rho}{=} \ad_B^\ast \Psi$.
 
 An orbit $\mathcal{O}$ equipped with a non-degenerate closed $2$-form \eqref{2form}
serves as a homogeneous symplectic manifold. The form $\varpi$ realizes the
isomorphism between the space of $1$-forms $\mathfrak{g}^\ast$ and the space of vector fields $T(\mathcal{O})$.
Below we explain this in more detail. 

Let $X_i$ form a basis of $\mathfrak{g}$,
and $\Xi_i$ form the dual basis of $\mathfrak{g}^\ast$ such that $\langle \Xi_i, X_j \rangle = \delta_{i,j}$,
where $\delta_{i,j}$ is the Kronecker delta. Then $\langle \Psi, X_i \rangle = \psi_i$ serve as coordinates 
of $\Psi \in \mathfrak{g}^\ast$ in the chosen basis. Let $\mathcal{C}(\mathfrak{g}^\ast)$
be the space of smooth functions on $\mathfrak{g}^\ast$. Let
$B(\mathfrak{g}^\ast, \mathfrak{g})$ be the space of closed $1$-forms $\nabla \mathcal{F} $ assigned to
 $\mathcal{F} \in \mathcal{C}(\mathfrak{g}^\ast)$ by the rule
$$\nabla \mathcal{F}(\Psi) = \sum_i \frac{\partial \mathcal{F}(\Psi)}{\partial \psi_i} X_i. $$
On the other hand, every $X \in \mathfrak{g} /  \mathfrak{s}_{\Psi}$ gives rise to a vector field
$\widetilde{X} \overset{\rho}{=} \ad^\ast_X \Psi$, tangent to $\mathcal{O} \subset \mathfrak{g}^\ast$.
Therefore, we have
\begin{gather*}
\mathcal{C}(\mathfrak{g}^\ast) \to  B(\mathfrak{g}^\ast, \mathfrak{g}) \overset{\rho}{\to} T(\mathcal{O}),\\
\mathcal{F}(\Psi) \mapsto \nabla \mathcal{F}(\Psi) \mapsto \ad^\ast_{\nabla \mathcal{F}(\Psi)} \Psi.
\end{gather*}
This defines the isomorphism between $1$-forms and vector fields.

The map $\rho$ is defined on all $1$-forms  $\sum_i  \mathcal{F}_i X_i \in \mathfrak{g}$.
The inverse map brings $\ad^\ast_{X(\Psi)} \Psi$ to $1$-form $\Xi_{\Psi}$ which acts on another vector field
$\ad^\ast_{A} \Psi$ as follows
$$ \Xi_{\Psi} (\ad^\ast_{A} \Psi) = \langle \Psi, [X(\Psi),A] \rangle.$$

In what follows, we call a function $\mathcal{F} \in \mathcal{C}(\mathfrak{g}^\ast)$
a hamiltonian, and  $\ad^\ast_{\nabla \mathcal{F}(\Psi)} \Psi$ 
the corresponding hamiltonian vector field. In particular, $\psi_i = \langle \Psi, X_i \rangle$
serve as  hamiltonians of the vector fields $\ad^\ast_{X_i} \Psi$.
A commutator of two hamiltonian vector fields is a hamiltonian vector field,
which is defined through the Lie---Poisson bracket:
\begin{gather}\label{LiePoiBra}
\begin{split}
\{\mathcal{F},\mathcal{H}\} &= \langle \Psi, \bigg[\sum_i  \frac{\partial \mathcal{F}(\Psi)}{\partial \psi_i} X_i,
\sum_j  \frac{\partial \mathcal{H}(\Psi)}{\partial \psi_j} X_j \bigg] \rangle \\
& =  \sum_{i,j} \frac{\partial \mathcal{F}(\Psi)}{\partial \psi_i} \frac{\partial \mathcal{H}(\Psi)}{\partial \psi_j}
\langle \Psi, [X_i,X_j] \rangle .
\end{split}
\end{gather}

Every hamiltonian $\mathcal{H}$ gives rise to a flow
\begin{subequations}%\label{HamFlow}
\begin{gather*}
\frac{\rmd \Psi}{\rmd \tau} = \ad^\ast_{\nabla \mathcal{H}} \Psi,
\end{gather*}
or in coordinates
\begin{gather*}%\label{HamFlowCoord}
\frac{\rmd \psi_i}{\rmd \tau} = \langle \ad^\ast_{\nabla \mathcal{H}} \Psi, X_i \rangle
= \{ \mathcal{H}, \psi_i\},
\end{gather*}
\end{subequations}
which is a system of hamiltonian equations, and $\tau$ serves as a parameter along the 
hamiltonian vector field. The fact that
$X_i \in \mathfrak{g}/\mathfrak{s}_\Psi$ is not essential, since $\ad^\ast_{X_i} \Psi =0$
if $X_i \in \mathfrak{s}_\Psi$.

%======================================
\section{Integrable systems on coadjoint orbits of a loop group}\label{s:KdVHier}
The KdV equation arises within the hierarchy of integrable Hamiltonian systems 
on coadjoint orbits of loop $\mathrm{SL}(2,\Real)$-group. 
Here we briefly recall this scheme, as presented in \cite{Holod82}
and recalled in \cite{BerHol07}.
Such a construction is based on the results of \cite{Adler78,AdlMoer80}.

Let $\widetilde{\mathfrak{g}} = \mathfrak{g} \otimes
\mathcal{P}(z,z^{-1})$, where $\mathcal{P}(z,z^{-1})$ denotes the algebra of Laurent series in~$z$,
and $\mathfrak{g}=\mathfrak{sl}(2,\Complex)$ has the standard basis
\begin{gather*}
\H = \frac{1}{2} \begin{pmatrix} 1 & 0 \\ 0 & -1  \end{pmatrix},\quad
\X = \begin{pmatrix} 0 & 1 \\ 0 & 0  \end{pmatrix},\quad
\Y = \begin{pmatrix} 0 & 0 \\ 1 & 0  \end{pmatrix}.
\end{gather*}
In the algebra $\widetilde{\mathfrak{g}}$ the \emph{principal grading} is introduced, defined by the grading operator 
$$ \mathfrak{d} = 2 z \frac{\rmd }{\rmd z}  + \ad_\H,$$
where $\ad$ denotes the adjoint operator in $\mathfrak{g}$,
that is $\forall\, \X\in \mathfrak{g}$ $\ad_\H \X = [\H,\X] $. 

Let $\{\X_{2m-1},\, \Y_{2m-1},\, \H_{2m},\mid m \in \Integer\}$, such that
\begin{gather*}
\X_{2m-1} = z^{m-1} \X,\quad
\Y_{2m-1} = z^{m} \Y,\quad \H_{2m} = z^m \H,
\end{gather*}
form a basis of $\widetilde{\mathfrak{g}}$. In general, an element of the basis will be denoted by $\Z_{a,\ell}$,
where $a=1$, $2$, $3$, and $\ell$ indicates the degree of $\Z_{a,\ell}$, namely $\mathfrak{d} \Z_{a,\ell} = \ell \Z_{a,\ell}$.
Actually, $\Z_{1,\ell}=\H_\ell$, $\Z_{2,\ell}=\Y_\ell$, $\Z_{3,\ell}=\X_\ell$.
Let $\mathfrak{g}_{\ell}$ denote the degree $\ell$ eigenspace of $\mathfrak{d}$.
Thus, $\mathfrak{g}_{2m-1} = \spanOp \{\X_{2m-1},\, \Y_{2m-1}\}$, and $\mathfrak{g}_{2m} = \spanOp \{\H_{2m}\}$.

According to the scheme from \cite{AdlMoer80}, $\widetilde{\mathfrak{g}}$ is divided into two
subalgebras 
\begin{gather*}
\widetilde{\mathfrak{g}}_+ = \oplus \sum_{\ell \geqslant 0} \mathfrak{g}_{\ell},\qquad
\widetilde{\mathfrak{g}}_- = \oplus \sum_{\ell < 0} \mathfrak{g}_{\ell}.
\end{gather*}

The bilinear form
\begin{gather*}%\label{KdVBiLinF}
\forall A(z), B(z) \in \widetilde{\mathfrak{g}} \qquad 
\langle A(z), B(z) \rangle = \res_{z=0} \tr A(z) B(z)
\end{gather*}
introduces the duality
\begin{gather*}
\X_{2m-1}^\ast \leftrightarrow \X_{-2m-1},\quad
\Y_{2m-1}^\ast \leftrightarrow  \Y_{-2m-1},\quad
\H_{2m}^\ast \leftrightarrow \H_{-2m-2},
\end{gather*}
where $\{\X_{2m-1}^\ast,\, \Y_{2m-1}^\ast,\, \H_{2m}^\ast,\mid m \in \Integer\}$,  such that
\begin{gather*}
\X_{2m-1}^\ast = z^{m} \X^\ast,\quad
\Y_{2m-1}^\ast = z^{m-1} \Y^\ast,\quad
\H_{2m}^\ast = z^m \H^\ast,
\end{gather*}
form the basis of $\widetilde{\mathfrak{g}}^\ast$,
and the basis elements of $\mathfrak{g}^\ast$ are
\begin{gather*}
\H^\ast = \begin{pmatrix} 1 & 0 \\ 0 & -1  \end{pmatrix},\quad
\X^\ast = \begin{pmatrix} 0 & 0 \\ 1 & 0  \end{pmatrix},\quad
\Y^\ast = \begin{pmatrix} 0 & 1 \\ 0 & 0  \end{pmatrix}.
\end{gather*}

Then the dual subalgebras $\widetilde{\mathfrak{g}}_+^\ast$ and $\widetilde{\mathfrak{g}}_-^\ast$ are
\begin{gather*}
\widetilde{\mathfrak{g}}_+^\ast =  \oplus \sum_{\ell \leqslant -2} \mathfrak{g}_{\ell},\qquad
\widetilde{\mathfrak{g}}_-^\ast = \widetilde{\mathfrak{g}}_+ \oplus \mathfrak{g}_{-1}.
\end{gather*}
Note, that $\mathfrak{g}_{-1}^\ast = \mathfrak{g}_{-1}$, 
and $\mathfrak{g}_{\ell}^\ast$ is dual to $\mathfrak{g}_{-\ell-2}$,
if $\ell > -1$.

%-----------------------------------------------
\subsection{The phase space of KdV hierarchy}
The KdV equation arises within the hierarchy of hamiltonian systems
on coadjoint orbits of the group $\widetilde{G}_- = \exp(\widetilde{\mathfrak{g}}_-)$.

Let 
$\M_N = \widetilde{\mathfrak{g}}^\ast_- / \big( \sum_{\ell \geqslant 2N+2}  \mathfrak{g}_\ell \big)$.
Actually,
\begin{gather*}%\label{MKdVPhSpGen}
\M_N = \bigg\{\Psi = \sum_{\ell=-1}^{2N+1}  \sum_{a=1,2,3}  \psi_{a,\ell} \Z_{a,\ell}^\ast \bigg\},
\end{gather*}
where $\psi_{a,\ell}$ serve as coordinates on $\M_N$, also called the \emph{dynamic variables},
$$\psi_{a,\ell} = \langle \Psi, \Z_{a,-\ell-2}\rangle.$$
Let $\psi_{1,\ell}=\alpha_{\ell}$, 
$\psi_{2,\ell} = \beta_{\ell}$, $\psi_{3,\ell} = \gamma_{\ell}$, then
every element $\Psi \in \M_N$ is a matrix polynomial in $z$ of the form
\begin{subequations}\label{MKdVPhSp}
\begin{gather}
\Psi(z) =
\begin{pmatrix} \alpha(z) & \beta(z) \\ \gamma(z) & -\alpha(z)
\end{pmatrix},\\
\alpha(z) = \sum_{m=0}^{N} \alpha_{2m} z^m,\quad
\beta(z) = \sum_{m=0}^{N+1} \beta_{2m-1} z^{m-1},\quad
\gamma(z) =  \sum_{m=0}^{N+1} \gamma_{2m-1} z^{m}.
\end{gather}
\end{subequations}

The action of $\widetilde{G}_-$ splits $\M_N$ into orbits
$$\mathcal{O} = \{\Psi = \Ad^\ast_{g} \Phi \mid g\in \widetilde{G}_-\},\qquad \Phi \in \M_N.$$
Initial points $\Phi$ can be taken from the Weyl chamber of $\widetilde{G}_-$ in $\M_N$.
The Weyl chamber is spanned by $\H_{2m}^\ast$, $m=0$, \ldots, $N$. Thus,
initial points $\Phi$ are given by diagonal matrices.

According to the construction presented in subsection~\ref{ss:HSCO},
coadjoint orbits possess a symplectic structure, which remains the same within $\M_N$.
Let $\mathcal{F}$, $\mathcal{H}\in \mathcal{C}(\M_N)$, then \eqref{LiePoiBra} acquires the form
\begin{subequations}\label{LiePoiBraKdV}
\begin{gather}
\{\mathcal{F},\mathcal{H}\} = \sum_{i,j =-1}^{2N+1}  \sum_{a, b=1,2,3} 
W_{i,j}^{a,b} \frac{\partial \mathcal{F}}{\partial \psi_{a,i}} \frac{\partial \mathcal{H}}{\partial \psi_{b,j}},
\\ W_{i,j}^{a,b}  = \langle \Psi, [\Z_{a,-i-2},\Z_{a,-j-2}] \rangle.
\end{gather}
\end{subequations}
We call $\M_N$ the symplectic manifold. The orbits $\mathcal{O}$ which constitute $\M_N$
serve as phase spaces. So hamiltonian systems arise.

In terms of the dynamic variables
the symplectic structure \eqref{LiePoiBraKdV} is defined by
\begin{gather}\label{MKdVPoiBra}
\begin{split}
&\{\beta_{2m-1}, \alpha_{2n}\} = \beta_{2(n+m)+1},\\
&\{\gamma_{2m-1}, \alpha_{2n}\} = - \gamma_{2(n+m)+1},\\
&\{\beta_{2m-1}, \gamma_{2n-1}\} = - 2 \alpha_{2(m+n)},\quad 0 \leqslant m+n \leqslant N.
\end{split}
\end{gather}
If $m+n>N$, such Poisson brackets vanish. In particular, for all $\psi_{a,\ell}$
\begin{gather*}
\{\beta_{2N+1}, \psi_{a,\ell}\} = 0,\qquad
\{\gamma_{2N+1}, \psi_{a,\ell}\} = 0.
\end{gather*}
That is,  $\beta_{2N+1}$, and $\gamma_{2N+1}$
are constant, we assign $\beta_{2N+1} = \gamma_{2N+1} = \rmb$.

Physically meaningful hamiltonian systems arise when $\mathfrak{g}$ is 
one of the real forms of $\mathfrak{sl}(2, \Complex)$, namely $\mathfrak{sl}(2, \Real)$
or $\mathfrak{su}(2)$.

\begin{rem}
In $\M_N$ with $\Psi$ of the form \eqref{MKdVPhSp}, coadjoint orbits of $\widetilde{G}_{-}$
serve as finite phase spaces for the hierarchy of the $+$mKdV equation in the case of $\mathfrak{g} = \mathfrak{su}(2)$, 
and $-$mKdV  in the case of $\mathfrak{g} = \mathfrak{sl}(2, \Real)$. At the same time,
coadjoint orbits of $\widetilde{G}_{+} = \exp(\widetilde{\mathfrak{g}}_+)$ serve as 
phase spaces for the hierarchies of the $\sin$-Gordon or  $\sinh$-Gordon equation,
respectively.
\end{rem}

The KdV hierarchy  is obtained in the case of $\mathfrak{g} = \mathfrak{sl}(2, \Real)$ by 
means of the hamiltonian reduction
\begin{gather}\label{KdVRed}
\mathfrak{g}_{-1} = \spanOp\{\X_{-1} + \Y_{-1}\},\qquad  \beta_{-1}=0.
\end{gather}
Let $\M_N^\circ$ denote $\M_N$ with this reduction applied.
On $\M_N^\circ$ \eqref{MKdVPoiBra} changes into
\begin{gather}\label{KdVPoiBra}
\begin{split}
&\{\gamma_{-1}, \alpha_{2n}\} = \beta_{2n+1} - \gamma_{2n+1}, \quad 0 \leqslant  n < N; \\
&\{\gamma_{2m-1}, \alpha_{2n}\} = - \gamma_{2(m+n)+1}, \\
&\{\beta_{2m-1}, \alpha_{2n}\} = \beta_{2(m+n)+1}, \\
&\{\beta_{2m-1}, \gamma_{2n-1}\} = - 2 \alpha_{2(m+n)},\quad 1 \leqslant m+n \leqslant N.
\end{split}
\end{gather}
Thus, $\{\alpha_{2N}, \psi_{a,\ell}\} = 0$ for all $\psi_{a,\ell}$, and so $\alpha_{2N}$ is constant, 
we assign $\alpha_{2N}= \rma$.

In what follows, we consider the KdV hierarchy only.
Let the dynamic variables on $\M_N^\circ$, of number $3N+1$, be ordered as follows: 
\begin{equation}\label{DynVar}
\gamma_{-1}, \{\alpha_{2m-2},\, \beta_{2m-1}, \gamma_{2m-1}\}_{m=1}^{N},
\end{equation}
and $\alpha_{2N}= \rma$, $\beta_{2N+1} = \gamma_{2N+1} = \rmb$ are constant. 

The Poisson structure $\textsf{W}= (W_{i,j}^{a,b})$ has the form
\begin{align*}
&\textsf{W} = \begin{pmatrix} 0 & \overset{\circ}{\w}_1 & \overset{\circ}{\w}_2 &
\dots & \overset{\circ}{\w}_{N-1} & \overset{\circ}{\w}_N \\
- \overset{\circ}{\w}{}_1^t& \w_2 &  \w_3 & \dots & \w_{N} & \w_{N+1} \\
- \overset{\circ}{\w}{}_2^t& \w_3 &  \w_4 & \dots & \w_{N+1} & 0\\
 \vdots & \vdots & \vdots & \iddots & 0 & 0 \\
- \overset{\circ}{\w}_{N-1}^t & \w_{N} & \w_{N+1} &  \iddots & \vdots & \vdots \\
- \overset{\circ}{\w}_{N}^t  & \w_{N+1}  & 0 &  \dots & 0 & 0
   \end{pmatrix},\\
&\w_n =  \begin{pmatrix} 0 & - \beta_{2n-1} & \gamma_{2n-1}  \\
\beta_{2n-1} & 0 & -2 \alpha_{2n} \\
 - \gamma_{2n-1} & 2 \alpha_{2n}  & 0 \\
 \end{pmatrix},\quad n = 1,\,\dots,\, N, \\
&\w_{N+1} =  \begin{pmatrix} 0 & - \rmb & \rmb  \\
\rmb & 0 & 0 \\
 - \rmb & 0 & 0 \\
 \end{pmatrix}, \\
&\overset{\circ}{\w}_n = \big(\beta_{2n-1} - \gamma_{2n-1},\, 2\alpha_{2n},\, -2\alpha_{2n}\big).
\end{align*}

%-----------------------------------------------
\subsection{Integrals of motion}
Invariant functions in the dynamic variables arise from 
\begin{equation}\label{InvF}
H(z) = \tfrac{1}{2} \tr \Psi^2(z) = \alpha(z)^2 + \beta(z) \gamma(z).
\end{equation}
The polynomial $H$ is of the form 
$$H(z) = h_{2N+1} z^{2N+1} + \cdots + h_1 z + h_0,$$ 
where
\begin{gather}\label{hExprs}
\begin{split}
& h_{2N+1} = \rmb^2,\\
& h_{2N} = \rma^2 + \rmb (\beta_{2N-1} + \gamma_{2N-1}),\\
& h_{2N-1} = 2 \rma \alpha_{2N-2} + \rmb (\beta_{2N-3} + \gamma_{2N-3}) + \beta_{2N-1} \gamma_{2N-1}, \\
&\dots \\
& h_{0} = \alpha_0^2 + \beta_1 \gamma_{-1}.
\end{split}
\end{gather}

Evidently, any evolution of $\Psi$ preserves $H(z)$.
Therefore, every $h_n$ serves as an integral of motion, and $h_{2N+1} = \rmb^2$ is an absolute constant. 
With respect to the symplectic structure \eqref{KdVPoiBra},
$h_0$, \ldots $h_{N-1}$ give rise to non-trivial hamiltonian flows, we call them hamiltonians. 
\begin{rem}
Within the $-$MKdV hierarchy, there exists one more hamiltonian $h_{-1} = \beta_{-1} \gamma_{-1}$,
which vanishes due to $\beta_{-1} = 0$ in the KdV hierarchy. Thus,  
$\M_N^\circ$  is  the surface of level $h_{-1}=0$ in $\M_N$. That is why this reduction is called hamiltonian.
\end{rem}

On the other hand, $h_N$, \ldots, $h_{2N}$ annihilate 
the Poisson bracket \eqref{KdVPoiBra}, since 
$$\sum_{i =-1}^{2N+1}  \sum_{a =1,2,3} 
W_{i,j}^{a,b} \frac{\partial h_n}{\partial \psi_{a,i}} = 0,  \quad n=N,\dots, 2N,$$ 
and so $\{h_n,\psi \} = 0$
 for any dynamic variable $\psi$. 
Thus, 
\begin{gather}\label{KdVConstr}
h_{N} = c_N,\quad \ldots,\quad  h_{2N} = c_{2N}
\end{gather}
serve as constraints on the symplectic manifold $\M_N^\circ$. These constraints
fix an orbit~$\mathcal{O}$ of dimension $2N$, which  serves 
as a finite phase space of a hamiltonian system.
The Poisson bracket \eqref{KdVPoiBra} is degenerate, and not canonical.
Further, we find canonical coordinates on each orbit of $\M_N^\circ$, 
they provide separation of variables.

%-----------------------------------------------
\subsection{KdV equation}
On $\M_N^\circ$ we consider two hamiltonians: $h_{N-1}$ gives rise to a stationary flow with parameter $\x$,
and $h_{N-2}$ gives rise to an evolutionary flow with parameter $\rmt$:
\begin{gather}\label{TwoFlowsEq}
\frac{\rmd \psi_{a,\ell}}{\rmd \x} = \{h_{N-1}, \psi_{a,\ell}\},\qquad
\frac{\rmd \psi_{a,\ell}}{\rmd \rmt} = \{h_{N-2}, \psi_{a,\ell}\}.
\end{gather}
In more detail, the stationary flow is
\begin{align*}
&\frac{\rmd \gamma_{-1}}{\rmd \x} = - 2 \alpha_{0} (\beta_{2N-1} - \gamma_{2N-1})
- 2 \rma \gamma_{-1},\\
&\frac{\rmd \alpha_{2m}}{\rmd \x} =  -\rmb (\beta_{2m-1} - \gamma_{2m-1}) 
+ \beta_{2m+1}(\beta_{2N-1} - \gamma_{2N-1}),\\
&\frac{\rmd \beta_{2m+1}}{\rmd \x} = - 2 \rmb \alpha_{2m} + 2 \rma \beta_{2m+1},\\
&\frac{\rmd \gamma_{2m+1}}{\rmd \x} = 2 \rmb \alpha_{2m} - 2 \alpha_{2m+2} (\beta_{2N-1} - \gamma_{2N-1})
- 2 \rma \gamma_{2m+1} , \quad m=0,\dots, N-1.
\end{align*}
From the evolutionary flow we are interested in the equation
\begin{align*}%\label{KdVEqZeroCurv}
&\frac{\rmd \beta_{2N-1}}{\rmd \rmt} = - 2 \rmb \alpha_{2N-4} + 2 \rma \beta_{2N-3}.
\end{align*}
Note that,
\begin{gather}\label{KdVEqZeroCurv}
\frac{\rmd \beta_{2N-1}}{\rmd \rmt} = \frac{\rmd \beta_{2N-3}}{\rmd \x},
\end{gather}
and this equality produces the KdV equation for the dynamic variable $\beta_{2N-1}$.
Indeed, eliminating $\gamma_{2N-3}$ from
$h_{2N-1} = c_{2N-1}$ and
$$ \frac{\rmd \alpha_{2N-2}}{\rmd \x} =  -\rmb (\beta_{2N-3} - \gamma_{2N-3}) 
+ \beta_{2N-1}(\beta_{2N-1} - \gamma_{2N-1}), $$
we find
\begin{equation}\label{Aux1} 
2 \rmb \beta_{2N-3}  = - \frac{\rmd  \alpha_{2N-2}}{\rmd \x}
+ \beta_{2N-1}(\beta_{2N-1} - 2 \gamma_{2N-1})
-2 \rma \alpha_{2N-2} + c_{2N-1}. 
\end{equation}
Then, $\alpha_{2N-2}$ is obtained from
$$ \frac{\rmd \beta_{2N-1}}{\rmd \x} = - 2 \rmb \alpha_{2N-2} + 2 \rma \beta_{2N-1},  $$
and $\gamma_{2N-1}$ from $h_{2N} = c_{2N}$. As a result, $\alpha_{2N-2}$ and
$\gamma_{2N-1}$ are expressed in terms of $ \beta_{2N-1}$ and its derivatives:
\begin{align*}
&\alpha_{2N-2} = - \frac{1}{2\rmb} \frac{\rmd \beta_{2N-1}}{\rmd \x} + \frac{\rma}{\rmb} \beta_{2N-1}, \\
&\gamma_{2N-1} = \frac{c_{2N} - \rma^2}{\rmb} - \beta_{2N-1}.
\end{align*}
Substituting these expressions into \eqref{Aux1}, we find
\begin{equation}\label{Aux2} 
 \beta_{2N-3}  =  \frac{1}{4 \rmb^2 } \Big( \frac{\rmd^2 \beta}{\rmd \x^2} 
+ 6 \rmb \beta^2 - 4 c_{2N} \beta + 2  c_{2N-1} \rmb \Big),
\end{equation}
where $\beta$ stands for $\beta_{2N-1}$.
Finally, differentiating \eqref{Aux2} with respect to $\x$ and substituting into \eqref{KdVEqZeroCurv},
we obtain 
\begin{equation}\label{KdVEqBeta} 
\frac{\rmd \beta}{\rmd \rmt} =  \frac{1}{4 \rmb^2 } \Big( \frac{\rmd^3 \beta}{\rmd \x^3} 
+ 12 \rmb \beta \frac{\rmd \beta}{\rmd \x}  - 4 c_{2N} \frac{\rmd \beta}{\rmd \x}  \Big),
\end{equation}
which is the KdV equation in the most general form. According to \cite{KdV},
$c_{2N}$ is in close connection with the velocity of the uniform motion given to the liquid.
In the conventional KdV equation this term is eliminated. 
By assigning $4\rmb^2 = 1$, $c_{2N} = 0$, we come to \eqref{KdVEq}.

\begin{rem}
Note, that the KdV equation arises on $\M_N^\circ$ with $N \geqslant 2$. On $\M_1$ there exists
only stationary flow, in which we have
\begin{subequations}
\begin{align}
&\frac{\rmd \alpha_0}{\rmd \x} =  \beta_1^2 + \rmb \gamma_{-1} - \beta_1 \gamma_1, \label{M1Eq1} \\
&\frac{\rmd \beta_1}{\rmd \x} =  - 2 \rmb \alpha_0 + 2 \rma \beta_1. \label{M1Eq2} 
\end{align}
\end{subequations}
From \eqref{M1Eq1} and $2 \rma \alpha_0 + \rmb \gamma_{-1} + \beta_1 \gamma_1 = c_1$ we eliminate
$\gamma_{-1}$. Then we find $\alpha_0$ from \eqref{M1Eq2}, and substitute into the former equation. Finally, we 
find $\gamma_1$ from $\rma^2 + \rmb (\beta_1 + \gamma_1) = c_2$, and substitute. As a result, we find
\begin{equation}\label{KdVEqDegen}
 \frac{\rmd^2 \beta_1}{\rmd \x^2} + 6 \rmb \beta_1^2 - 4 c_2 \beta_1 = - 2 \rmb c_1,
\end{equation}
which is the first integral of the stationary KdV equation for $\beta_1 \equiv \beta$.

\end{rem}

%-----------------------------------------------
\subsection{Higher KdV equations}
On $\M_N^\circ$ with $N > 2$ one has $N-2$ higher KdV equations, which come from
\begin{gather*}
\frac{\rmd \beta_{2N-1}}{\rmd \tau_{m-1}} = \frac{\rmd \beta_{2m-1}}{\rmd \x},\quad  m = 1,\, \dots,\, N-2,
\end{gather*}
where $\tau_{n}$ denotes a parameter of the flow $h_{n}$.

For example, the first higher KdV equation arises when $m=N-2$, and has the form
($c_{2N}=0$, $\beta_{2N-1} \equiv \beta$)
\begin{multline*}%\label{H1KdVEqBeta} 
\frac{\rmd \beta}{\rmd \tau_{N-3}} =  \frac{1}{16 \rmb^4 } \bigg( \frac{\rmd^5 \beta}{\rmd \x^5} 
+ 20 \rmb \Big(\beta \frac{\rmd^3 \beta}{\rmd \x^3} + 	2 \frac{\rmd \beta}{\rmd \x} \frac{\rmd^2 \beta}{\rmd \x^2}\Big) 
+ 120 \rmb^2 \beta^2 \frac{\rmd \beta}{\rmd \x} \\  + 8 \rmb^2 c_{2N-1} \frac{\rmd \beta}{\rmd \x}  \bigg),
\end{multline*}
which coincides with \cite[Eq.\,(8''), p.\,244]{Nov1974}.

%-----------------------------------------------
\subsection{Zero curvature representation}
The system of dynamical equations \eqref{TwoFlowsEq} admits  the matrix form
\begin{gather*}
\frac{\rmd \Psi}{\partial \x} = [\Psi, \nabla h_{N-1}], \qquad
\frac{\rmd \Psi}{\partial \rmt} = [\Psi, \nabla h_{N-2}],
\end{gather*}
where $\nabla h_n$ denotes the matrix gradient of $h_n$, namely,
\begin{gather*}
\nabla h_n = \sum_{i =-1}^{2N+1}  \sum_{a =1,2,3} 
\frac{\partial h_n}{\partial \psi_{a,i}} \Z_{a,-i-2}.
\end{gather*}
The matrix gradient of each flow has a complementary matrix $\mathrm{A}$, such that
$$[\Psi, \nabla h_{n}] =  [\Psi,  \mathrm{A}].$$  Unlike $\nabla h_n$, 
the complementary matrix $A$ is defined in the same way in all $\M^\circ_{\tilde{N}}$, $\tilde{N} \geqslant N$.
Actually,
%\begin{subequations}
\begin{gather}\label{PsiMatrEqs} 
\frac{\rmd \Psi}{\partial \x} = [\Psi,  \mathrm{A}_\text{st}], \qquad
\frac{\rmd \Psi}{\partial \rmt} = [\Psi,  \mathrm{A}_\text{ev}], \\
\begin{split}
&\mathrm{A}_\text{st} = - \begin{pmatrix} 
\rma & \rmb \\  
\rmb z + \gamma_{2N-1}  - \beta_{2N-1} & - \rma
 \end{pmatrix},\\
&\mathrm{A}_\text{ev} = - \begin{pmatrix} 
\rma z + \alpha_{2N-2} & \rmb z + \beta_{2N-1} \\  
\rmb z^2 + \gamma_{2N-1} z + \gamma_{2N-3}  - \beta_{2N-3} & - (\rma z + \alpha_{2N-2})
 \end{pmatrix}.
 \end{split} \notag
\end{gather}
%\end{subequations}

The zero curvature representation for the KdV hierarchy has the form
\begin{gather*}
\frac{\rmd \mathrm{A}_\text{st}}{\partial \rmt} - \frac{\rmd \mathrm{A}_\text{ev}}{\partial \x} 
= [\mathrm{A}_\text{st}, \mathrm{A}_\text{ev}].
\end{gather*}

\subsection{Summary}
The affine algebra $ \widetilde{\mathfrak{g}} = \mathfrak{sl}(2,\Real) \otimes \mathcal{P}(z,z^{-1})$
with the \emph{principal grading} is associated with the KdV hierarchy.
Let $\M^\circ_N$ be the  manifold 
$\widetilde{\mathfrak{g}}^\ast_- / \big( \sum_{\ell \geqslant 2N+2}  \mathfrak{g}_\ell \big)$
with the hamiltonian reduction \eqref{KdVRed}, $\dim \M^\circ_N = 3N+1$, $N\,{\in}\, \Natural$.
Evidently,  $\M^\circ_1 \subset \M^\circ_2 \subset \cdots \subset 
\M^\circ_N \subset \M^\circ_{N+1} \subset \cdots $.
Each manifold $\M^\circ_N$ is equipped
with the symplectic structure \eqref{KdVPoiBra}.
Under the action of the loop group $\widetilde{G}_- = \exp(\widetilde{\mathfrak{g}}_-)$
a manifold $\M^\circ_N$ splits into orbits $\mathcal{O}$, each generated by a point 
from the Weyl chamber. On the other hand, such an orbit is defined by the system of $N+1$ constraints
\eqref{KdVConstr}. Each orbit serves as a phase space of dimension $2N$ for a hamiltonian system 
integrable in the Liouville sence.

On orbits within $\M^\circ_N$, $N \geqslant 2$, there exist two hamiltonians
whose flows give rise to the KdV equation \eqref{KdVEqBeta}. We call these flows
stationary and evolutionary with parameters $\x$ and $\rmt$, correspondingly.
On orbits in $\M^\circ_1$ there exists  a stationary flow only.
If $N>2$, higher KdV equations arise. One can use the remaining hamiltonians
to generate  evolutionary flows.

%======================================
\section{Separation of variables}\label{s:SoV}

%-----------------------------------------------
\subsection{Spectral curve}
The KdV hierarchy presented above is associated with the family of hyperelliptic curves
\begin{equation}\label{SpectrCurve}
- w^2 + H(z) = 0.
\end{equation}
Indeed, the spectral curve of each hamiltonian system in the hierarchy
is defined by the characteristic polynomial of $\Psi$, namely 
$\det \big(\Psi(z) - w\big) = 0$. Recall, that  $H$ is the polynomial \eqref{InvF} of degree $2N+1$.
All coefficients are integrals of motion: $h_0$, \ldots, $h_{N-1}$ serve as hamiltonians,
and $h_N$, \dots, $h_{2N}$ introduce constraints \eqref{KdVConstr}, which fix an orbit, 
and $h_{2N+1} = \rmb^2$.

%-----------------------------------------------
\subsection{Canonical coordinates}
As shown in \cite{KK1976}, variables of separation in the hierarchy of the $\sin$-Gordon equation
are given by certain points of a spectral curve in each system of the hierarchy.
In \cite{BerHol07}, this result was extended to 
all integrable systems with spectral curves from the hyperelliptic family.
In general, pairs of coordinates of a certain number of points serve as quasi-canonical variables, 
and so lead to separation of variables.
Below, we briefly explain how to find the required points in the KdV hierarchy, and prove that  
pairs of coordinates of these $N$ points serve as canonical variables on $\M_N^\circ$.

Recall, that the symplectic manifold $\M_N^\circ$ is described by $3N+1$
dynamic variables \eqref{DynVar}. At the same time, 
each orbit $\mathcal{O}$ in $\M_N^\circ$  is fixed by $N+1$ constraints, and so $\dim \mathcal{O}= 2N$.
Thus, $N+1$ dynamic variables can be eliminated with the help of these constraints.
We eliminate variables $\gamma_{2m-1}$, $m=0$, \ldots, $N$. Note, that
all expressions \eqref{hExprs} are linear with respect to $\gamma_{2m-1}$.
The constraints, together with $h_{2N+1} = \rmb^2$, in the matrix form are
\begin{gather*}%\label{Constr}
\B_{\text{c}} \bm{\gamma} + \A_{\text{c}} = \bm{c},\\
\intertext{where}
\B_{\text{c}} = \begin{pmatrix}
\rmb & 0 &\ddots & 0 & 0 & 0 \\
\beta_{2N-1} & \rmb & \ddots & 0 & 0 & 0 \\
\vdots & \beta_{2N-1} & \ddots & 0 & 0 & 0 \\
\beta_3 & \vdots & \ddots &  \rmb & 0 & 0 \\
\beta_1 & \beta_3 & \dots  & \beta_{2N-1} & \rmb & 0 \\
0 & \beta_1 & \beta_3 & \dots  & \beta_{2N-1} & \rmb 
\end{pmatrix},\qquad
\bm{\gamma} = \begin{pmatrix} \rmb \\
\gamma_{2N-1} \\ \gamma_{2N-3} \\ \vdots \\ 
\gamma_3 \\ \gamma_1 \\ \gamma_{-1}
\end{pmatrix},  \notag \\
\A_{\text{c}} = \begin{pmatrix} 0 \\ \rma^2 \\
2 \rma \alpha_{2N-2} \\ \vdots \\ 
2\rma \alpha_2 + \sum_{n=1}^{N-2} \alpha_{2(N-n)} \alpha_{2(n+1)} \\ 
2\rma \alpha_0 + \sum_{n=1}^{N-1} \alpha_{2(N-n)} \alpha_{2n} \end{pmatrix},\qquad
\bm{c} = \begin{pmatrix}  \rmb^2 \\ c_{2N} \\ c_{2N-1} \\ \vdots \\ c_{N+1} \\ c_N \end{pmatrix}.  \notag
\end{gather*}
The first equation is an identity, we include it to make the matrix $\B$ square and invertible. 
Then
\begin{equation}\label{GammaFromConstr}
\bm{\gamma} = \B_{\text{c}}^{-1}(\bm{c} - \A_{\text{c}}).
\end{equation}

The remaining expressions, which represent hamiltonians $h_0$, \ldots, $h_{N-1}$,
have the matrix form
\begin{gather}\label{Hams}
\B_{\text{h}} \bm{\gamma} + \A_{\text{h}} = \bm{h},\\
\intertext{where}
\B_{\text{h}} = \begin{pmatrix}
0 & 0& \beta_{1} & \beta_3 &\dots & \beta_{2N-3} & \beta_{2N-1} \\
0 & 0 & 0 & \beta_1 & \beta_3 &\dots & \beta_{2N-3} \\
\vdots & \vdots & \vdots & \ddots & \ddots & \ddots & \vdots \\
0 & 0 & 0 & 0 & \dots & \beta_1 & \beta_3 \\
0 & 0 & 0 & 0 & \dots & 0 & \beta_1
\end{pmatrix},  \notag \\
\A_{\text{h}} = \begin{pmatrix} 
 \sum_{n=1}^{N} \alpha_{2(N-n)} \alpha_{2(n-1)} \\
 \sum_{n=2}^{N} \alpha_{2(N-n)} \alpha_{2(n-2)} \\ 
\vdots \\  2\alpha_0 \alpha_2 \\ \alpha_0^2 \end{pmatrix},\qquad
\bm{h} = \begin{pmatrix} h_{N-1} \\ h_{N-2} \\ \vdots \\ h_1 \\ h_0 \end{pmatrix}.  \notag
\end{gather}
Substituting \eqref{GammaFromConstr} into \eqref{Hams}, we obtain
\begin{equation}\label{HamsBeta}
 \bm{h} = \B_{\text{h}}  \B_{\text{c}}^{-1}(\bm{c} - \A_{\text{c}})+ \A_{\text{h}}.
\end{equation}

On the other hand, hamiltonians $h_0$, \ldots, $h_{N-1}$ can be found from 
the equation of the spectral curve, taken at $N$ points 
which form a non-special\footnote{Here a non-special divisor is supposed to be a positive divisor
of degree $g$ on a hyperelliptic curve of genus $g$ contains no pair of points in involution.} 
divisor. Namely, with $i=1$, \ldots, $N$
\begin{equation*}
 - w_i^2 + \rmb^2 z_i^{2N+1} + c_{2N} z_i^{2N} + \cdots + c_N z_i^N 
 + h_{N-1} z_i^{N-1} + \dots + h_1 z_i + h_0 = 0,
\end{equation*}
or in the matrix form
\begin{equation*}%\label{SpectrCurveEqs}
- \bm{w} +  \mathrm{Z}_{\text{c}} \bm{c} +  \mathrm{Z}_{\text{h}} \bm{h} = 0,
\end{equation*}
where
\begin{gather*}
\mathrm{Z}_{\text{c}} = \begin{pmatrix} 
z_1^{2N+1} & z_1^{2N} & \dots & z_1^{N} \\
z_2^{2N+1} & z_2^{2N} & \dots & z_2^{N} \\
\vdots &  \vdots & \ddots & \vdots \\
z_N^{2N+1} & z_N^{2N} & \dots & z_N^{N} 
\end{pmatrix},\quad
\mathrm{Z}_{\text{h}} = \begin{pmatrix} 
z_1^{N-1} & \dots & z_1 & 1 \\ 
z_2^{N-1} & \dots & z_2 & 1 \\ 
\vdots &  \ddots & \vdots & \vdots \\
z_N^{N-1} & \dots & z_N & 1 
\end{pmatrix},\quad 
\bm{w} =  \begin{pmatrix}
w_1^2 \\ w_2^2 \\ \vdots \\ w_N^2
\end{pmatrix}.
\end{gather*}
The matrix $\mathrm{Z}_{\text{h}}$ is square and invertible. Thus,
\begin{equation}\label{HamsSpectrC}
 \bm{h} = \mathrm{Z}_{\text{h}}^{-1} \big(\bm{w} -  \mathrm{Z}_{\text{c}} \bm{c}\big).
\end{equation}

Equations \eqref{HamsBeta} and \eqref{HamsSpectrC}
give the same hamiltonians. Therefore, 
\begin{gather*}
 \B_{\text{h}}  \B_{\text{c}}^{-1}(\bm{c} - \A_{\text{c}})+ \A_{\text{h}} = 
 \mathrm{Z}_{\text{h}}^{-1} \big(\bm{w} -  \mathrm{Z}_{\text{c}} \bm{c}\big).
\end{gather*}
Moreover, constants $\bm{c}$ can be taken arbitrarily,
and so we equate the corresponding coefficients, and the remaining terms:
\begin{subequations}
\begin{gather}
 \B_{\text{h}}  \B_{\text{c}}^{-1}  = - \mathrm{Z}_{\text{h}}^{-1}  \mathrm{Z}_{\text{c}}, \label{cCoefs} \\
  -\B_{\text{h}}  \B_{\text{c}}^{-1} \A_{\text{c}} + \A_{\text{h}} = \mathrm{Z}_{\text{h}}^{-1} \bm{w}. \label{fCoefs}
\end{gather}
\end{subequations}
From \eqref{cCoefs} we find
$$\mathrm{Z}_{\text{h}} \B_{\text{h}} + \mathrm{Z}_{\text{c}} \B_{\text{c}} = 0,$$
which is equivalent to $\beta(z_i) = 0$, since $\mathrm{Z}_{\text{h}} $ is the
Vandermonde matrix.
Then from \eqref{fCoefs} we obtain
$$ \mathrm{Z}_{\text{c}} \A_{\text{c}} + \mathrm{Z}_{\text{h}}\A_{\text{h}} = \bm{w}, $$
which is equivalent to
$w_i^2 - \alpha(z_i)^2 = 0$.
Thus, points $(z_i,w_i)$ are defined by
\begin{equation*}\label{PointsDef2}
\beta(z_i) = 0,\qquad w_i^2 - \alpha(z_i)^2 = 0,  \quad i=1,\dots, N.
\end{equation*}
A  similar result  was firstly discovered in 
 \cite{KK1976} regarding the hierarchy of the $\sin$-Gordon equation. 

\begin{theo}
Suppose an orbit $\mathcal{O} \subset \M_N^\circ$ 
has the coordinates $(\beta_{2m+1}, \alpha_{2m})$, $m=0$, \ldots, $N-1$, as above. 
Then the new coordinates $(z_i, w_i)$, $i=1$, \ldots, $N$, defined by the formulas
\begin{equation}\label{PointsDef1}
\beta(z_i) = 0,\qquad w_i = \epsilon \alpha(z_i),  \quad i=1,\dots, N,
\end{equation}
where $\epsilon^2 = 1$,
have the following properties:
\begin{enumerate}
\renewcommand{\labelenumi}{\arabic{enumi})}
\item  a pair $(z_i, w_i)$ is a root of the characteristic polynomial \eqref{SpectrCurve}.
\item a pair $(z_i, w_i)$ is canonically conjugate with respect to the  Lie-Poisson
bracket \eqref{KdVPoiBra}:
\begin{equation}\label{CanonCoord}
\{z_i, z_j\} = 0,\qquad  \{z_i, w_j\} = -\epsilon\, \delta_{i,j},\qquad
\{w_i, w_j\} = 0.
\end{equation}
\item the canonical $1$-form is
\begin{equation}\label{Liouv1Form}
 - \epsilon \sum_{i=1}^N w_i \rmd z_i.
\end{equation}
\end{enumerate}
\end{theo}
\begin{proof}
Since $z_i$ depend only on $\beta_{2m+1}$, $m=0$, \ldots, $N-1$, and the latter commute,
we have $\{z_i, z_j\} = 0$. Next,
\begin{multline*}
 \{z_i, w_j\}  = \sum_{n+m=0}^{N-1} \bigg(\frac{\partial z_i}{\beta_{2m+1}}  \frac{\partial w_j}{\alpha_{2n}}
 -  \frac{\partial z_i}{\alpha_{2n}}  \frac{\partial w_j}{\beta_{2m+1}} \bigg) \{\beta_{2m+1}, \alpha_{2n}\} \\
 = \frac{-\epsilon}{\beta'(z_i)} \sum_{m+n =0}^{N-1}  z_i^m z_j^n \beta_{2(m+n)+3}
 =  \frac{-\epsilon}{\beta'(z_i)}  \frac{\beta(z_i) - \beta(z_j)}{z_i - z_j},
\end{multline*}
since from \eqref{PointsDef1} we have 
$$\frac{\partial z_i}{\beta_{2m+1}}  = - \frac{z_i^m}{\beta'(z_i)},\qquad
 \frac{\partial z_i}{\alpha_{2n}} = 0,\qquad  
 \frac{\partial w_i}{\alpha_{2n}} = \epsilon z_i^n,\qquad
 \frac{\partial w_i}{\beta_{2m+1}} = - \epsilon z_i^m \frac{ \alpha'(z_i) }{\beta'(z_i)}.
$$
As $i \neq j$, it is evident that $\{z_i, w_j\}  = 0$, due to $\beta(z_i)=\beta(z_j)=0$.
As $i=j$, we get
$$  \{z_i, w_i\}  = \lim_{z_j \to z_i}  \frac{-\epsilon}{\beta'(z_i)}  \frac{\beta(z_i) - \beta(z_j)}{z_i - z_j} 
= -\epsilon. $$
Finally, we find
\begin{multline*}
 \{w_i, w_j\}  = \sum_{n+m=0}^{N-1} \bigg(\frac{\partial w_i}{\beta_{2m+1}}  \frac{\partial w_j}{\alpha_{2n}}
 -  \frac{\partial w_i}{\alpha_{2n}}  \frac{\partial w_j}{\beta_{2m+1}} \bigg) \{\beta_{2m+1}, \alpha_{2n}\} \\
 = - \epsilon^2 \sum_{m+n =0}^{N-1}  \bigg(z_i^m z_j^n \frac{ \alpha'(z_i) }{\beta'(z_i)}  
 - z_i^n z_j^m \frac{ \alpha'(z_j) }{\beta'(z_j)} \bigg) \beta_{2(m+n)+3} \\
 =  - \epsilon^2 \frac{\beta(z_i) - \beta(z_j)}{z_i - z_j} 
 \bigg(\frac{ \alpha'(z_i) }{\beta'(z_i)}  - \frac{ \alpha'(z_j) }{\beta'(z_j)}  \bigg).
\end{multline*}
Thus,  $\{w_i, w_j\}  = 0$,  due to $\beta(z_i)=\beta(z_j)=0$.

Then \eqref{Liouv1Form} follows from the fact that pairs $(z_i, w_i)$, $i=1$, \ldots, $N$, 
are canonically conjugate with \eqref{CanonCoord}.
\end{proof}

In what follows we assign $\epsilon = -1$.

%---------------------------------------
\subsection{Summary}
An orbit $\mathcal{O} \subset \M^\circ_N$, which serves as a phase space of dimension $2N$,
is completely parameterized by 
non-canonical variables  $\alpha_{2n}$, $\beta_{2n+1}$, $n=0$, \ldots, $N-1$.
The  variables $\gamma_{2n-1}$, $n=0$, \ldots, $N$, are eliminated 
with the help of the orbit equations \eqref{KdVConstr}. It is shown,
that $N$ points of the spectral curve \eqref{SpectrCurve}
chosen according to \eqref{PointsDef1} are canonical and serve as variables of separation.
In fact, these $N$ points give a solution of the Jacobi inversion problem \eqref{PointsDef1},
where the coefficients $\beta_{2n+1}$, $\alpha_{2n}$ of polynomials 
fix values of $\wp$-functions, and define a unique point within the fundamental domain of the
Jacobian variety of the spectral curve,
as we see below.

%======================================
\section{Algebro-geometric integration}\label{s:AGI}
%---------------------------------------
\subsection{Uniformization of the spectral curve}
After separation of variables, we came to the Jacobi inversion problem 
for a non-special divisor of $N$ points $\{(z_k, w_k)\}_{k=1}^N$
on a hyperelliptic curve of genus $N$
\begin{multline}\label{CurveGN}
 0 = F(z,w) \equiv - w^2 + \rmb^2 z^{2N+1} + c_{2N} z^{2N} + \cdots + c_N z^N \\
 + h_{N-1} z^{N-1} + \dots + h_1 z + h_0.
\end{multline}
Not normalized differentials  of the first and second kinds acquire the form
\begin{subequations}\label{K12Difs}
\begin{align}
& \rmd u_{2n-1} =  \frac{\rmb z^{N-n} \rmd z}{\partial_w F(z,w)},\quad n=1,\dots, g,\label{K1Difs} \\
& \rmd r_{2n-1} =  \frac{ \rmd z}{\rmb \partial_w F(z,w)} \sum_{j=1}^{2n-1} (2n-j) h_{2N+2-j} z^{N+n-j}, \label{K2Difs}
\end{align}
\end{subequations}
which follow from \eqref{K1DifsGen}, \eqref{K2DifsGen}, after reducing \eqref{CurveGN}
to the form \eqref{V22g1Eq} by applying the
transformation: $z\mapsto x$, $w \mapsto \rmb y$, $F(z,w) \mapsto f(x,y)= F(x, \rmb y)/\rmb^2$.

On the curve \eqref{CurveGN}, a solution
of the Jacobi inversion problem, that is
the Abel pre-image of $\mathcal{A}(D) = u \in \Jac$, with
a non-special positive divisor
 $D = \sum_{k=1}^N (z_k,w_k)$, is given by the system
\begin{subequations}\label{JIPb}
\begin{align}
&z^N - \sum_{k=1}^N z^{N-k} \wp_{1,2k-1}(u) = 0,\\
&\frac{2 w}{\rmb} + \sum_{k=1}^N z^{N-k} \wp_{1,1,2k-1}(u) = 0.
\end{align}
\end{subequations}

According to \eqref{PointsDef1}, the $N$ values $z_i$ are zeros of the polynomial
$\beta(z)$, and the $N$ values $w_i$ satisfy $w_i = - \alpha(z_i)$. Thus,
\begin{subequations}\label{AlphaBetaUni}
\begin{align}
&\beta_{2(N-k)+1} = - \rmb \wp_{1,2k-1}(u), \\
&\alpha_{2(N-k)} = \tfrac{1}{2} \rmb \wp_{1,1,2k-1}(u) - \rma \wp_{1,2k-1}(u), \quad k=1,\dots, N.
\end{align}
\end{subequations}

Therefore, a solution of  \eqref{KdVEqBeta} is 
\begin{equation}\label{KdVSol}
\beta \equiv \beta_{2N-1} = - \rmb \wp_{1,1}(u).
\end{equation}
This solution arose in \cite[Theorem 4.12]{belHKF}.

\begin{rem}
The fact, that \eqref{KdVSol} serves as a solution of the KdV equation
follows immediately from the relation
\begin{equation}\label{KdVIntWP}
- \wp_{1,1,1,1}(u) + 4 \wp_{1,3}(u) + 6 \wp_{1,1}(u)^2 + 4 \lambda_2 \wp_{1,1}(u) + 2 \lambda_4 =0,
\end{equation}
which holds for hyperelliptic $\wp$-functions in any genus. The relation
corresponds to a curve of the form \eqref{V22g1Eq}. Assigning
$\lambda_2 = c_{2N}/\rmb^2$, $\lambda_4 = c_{2N-1}/\rmb^2$, 
we get the relation for the spectral curve \eqref{CurveGN}.
Differentiation with respect to $u_1$ transforms \eqref{KdVIntWP}
into the KdV equaiton \eqref{KdVEqBeta} in terms of $\wp$-functions.

The relation \eqref{KdVIntWP} is well known in the elliptic case ($N=1$). 
In terms of the Weierstrass function $\wp(u;g_2,g_3)$
it acquires the form
$$- 2 \wp''(u) + 12 \wp(u)^2 - g_2 = 0,$$
where $\wp(u;g_2,g_3) \equiv \wp_{1,1}(u;-4\lambda_4,-4\lambda_6)$, 
$ \wp''(u;g_2,g_3) \equiv \wp_{1,1,1,1}(u;-4\lambda_4,-4\lambda_6)$, 
and $\wp_{1,3}(u)$ vanishes since $u$ has only one component in genus $1$.
The latter relation is obtained by differentiating  the equation
$$ (\wp'(u) )^2 = 4 \wp(u)^3 - g_2 \wp(u) - g_3.$$
Note, that the function $\wp_{1,1}$ introduced above corresponds to a curve of the form 
\eqref{V22g1Eq}, which contains one extra term with the coefficient $\lambda_2$, and 
$\lambda_4 = -g_2/4$, $\lambda_6 = -g_3/4 $.
\end{rem}

%-----------------------------------------------
\subsection{Equation of motion in variables of separation}
From \eqref{PsiMatrEqs} we find
\begin{subequations}
\begin{align*}
&\frac{\rmd}{\rmd \x} \beta(z) = 2 \rma \beta(z) - 2 \rmb \alpha(z),\\
&\frac{\rmd}{\rmd \rmt} \beta(z) = 2 (\rma z + \alpha_{2N-2}) \beta(z) - 2 (\rmb z + \beta_{2N-1}) \alpha(z),
\end{align*}
\end{subequations}
where all dynamic variables are functions of $\x$ and $\rmt$. Therefore, 
zeros of $\beta(z)$ are  functions of $\x$ and $\rmt$ as well, namely
 $\beta(z) = \rmb \prod_{k=1}^N (z-z_k(\x,\rmt))$.
Then
\begin{subequations}
\begin{align*}
&\frac{\rmd}{\rmd \x} \log \beta(z) = - \frac{1}{z-z_k}\frac{\rmd z_k}{\rmd \x}
= 2 \rma - 2 \rmb \frac{\alpha(z)}{\beta(z)},\quad k=1,\dots, N,\\
&\frac{\rmd}{\rmd \rmt} \log \beta(z) = - \frac{1}{z-z_k}\frac{\rmd z_k}{\rmd \rmt}
= 2 (\rma z + \alpha_{2N-2})  - 2 (\rmb z + \beta_{2N-1}) \frac{\alpha(z)}{\beta(z)}.
\end{align*}
\end{subequations}
Taking into account \eqref{PointsDef1}, we find as $z\to z_k$, $k=1$, \ldots, $N$,
%\begin{subequations}
\begin{align*}
&\frac{\rmd z_k}{\rmd \x}
= \frac{2 w_k}{\prod_{j\neq k}^N (z_k - z_j)},&
& \frac{\rmd z_k}{\rmd \rmt}
=  - \frac{2 w_k \sum_{j\neq k} z_j}{\prod_{j\neq k} (z_k - z_j) }.&
\end{align*}
%\end{subequations}

Now, let $D$ be a divisor of points $\{(z_k,w_k)\}_{k=1}^N$ defined by \eqref{PointsDef1}.
The Abel image
\begin{equation*}
u = \mathcal{A}(D) = \sum_{k=1}^N \int_{\infty}^{(z_k,w_k)} \rmd u = 
\sum_{k=1}^N \int_{\infty}^{(z_k,w_k)} 
\begin{pmatrix} 1\\ z \\ \vdots \\ z^{N-1} \end{pmatrix} 
\frac{\rmb \rmd z}{-2w}
\end{equation*}
depends on $\x$ and $\rmt$, since the points $(z_k,w_k)$ are functions of $\x$ and $\rmt$.
Then
\begin{subequations}
\begin{align*}
&\frac{\rmd u_{2n-1}}{\rmd \x} = \sum_{k=1}^N \frac{\rmb z_k^{N-n}}{-2w_k} \frac{\rmd z_k}{\rmd \x}
= \sum_{k=1}^N \frac{-\rmb z_k^{N-n}}{\prod_{j\neq k}^N (z_k - z_j)} = - \rmb \delta_{n,1}, \\
& \frac{\rmd u_{2n-1}}{\rmd \rmt} = \sum_{k=1}^N \frac{\rmb z_k^{N-n}}{-2w_k} \frac{\rmd z_k}{\rmd \rmt}
=  \sum_{k=1}^N  \frac{\rmb z_k^{N-n}\sum_{j\neq k} z_j}{\prod_{j\neq k} (z_k - z_j) } = - \rmb \delta_{n,2}.
\end{align*}
\end{subequations}
Thus, $u_{1} = - \rmb \x + C_1$, $u_3 = -\rmb \rmt + C_3$, and $u_{2n-1} = C_{2n-1} = \const$, $n=3$, \ldots, $N$.

Therefore, the finite-gap solution of the KdV equation \eqref{KdVEqBeta}  in 
the $2N$-dimensional phase space ($N>1$) is
\begin{equation}\label{KdVSolRealCond}
\beta(\x,\rmt) = - \rmb \wp_{1,1}(- \rmb \x + C_1,-\rmb \rmt + C_3, C_5, \dots, C_{2N-1}).
\end{equation}

Since $\rmb^2\in \Real$, we have two possibilities: (i) $\rmb$ is real, or (ii) $\rmb$ is purely imaginary.
In the case (i), the first two arguments $u_1$, $u_3$ of $\wp_{1,1}$ run along lines parallel to the real axes.
In the case (ii), the first two arguments $u_1$, $u_3$ of $\wp_{1,1}$ go  parallel to the imaginary axes.
If $\beta$ describes a quasi-periodic wave,
none of the mentioned lines coincides with the real or imaginary axis, due to the singularity of $\wp_{1,1}$ at $u=0$.

Next, we find  the constant vector $\bm{C} = (C_1, C_3, \dots, C_{2N-1})$ 
such that  $\beta$ in \eqref{KdVSolRealCond} is real-valued. We call this 
 the reality conditions.

%---------------------------------------
\subsection{Summary}
The uniformization of the spectral curve is given by \eqref{JIPb} in an implicit form.
On the other hand, it brings explicit expressions \eqref{AlphaBetaUni} 
for dynamic variables $\beta_{2n-1}$, $\alpha_{2n}$.  Coordinates $u_1$ and $u_3$
of the Jacobian variety of the spectral curve serve, up to a constant multiple $-\rmb$,
 as parameters $\x$ and $\rmt$ of the stationary and evolutionary flows, correspondingly.

%======================================
\section{Reality conditions}\label{s:RealCond}
It is known that an $N$-gap hamiltonian system in terms of variables of separation $(z_i,w_i)$ 
splits into $N$ independent systems.
Each system with coordinate $z_i$ and momenta $w_i$ describes a motion of mass $1/2$ in the potential $-H(z_i)$,
and periodic motion is located between a pair of roots of the potential where $H(z)>0$. Thus, finite trajectories
contain branch points as turning points, and so the argument of $\wp_{1,1}$ in \eqref{KdVSolRealCond} 
attains half-periods. In order to guarantee that $\wp_{1,1}$ is real-valued
we assume that all branch points are real.

Further,
we work with the spectral curve \eqref{SpectrCurve} reduced to the canonical form~\eqref{V22g1Eq}.

%-----------------------------------------------
\begin{prop}\label{P:ReImPeriods}
Let a hyperelliptic curve of genus $g$ have all real branch points. Then with
a choice of cycles as on  fig.~\ref{cyclesOdd} and the standard not normalized
holomorphic differentials \eqref{K1DifsGen}, all entries of the period matrix $\omega$ are real, and
all entries of the period matrix $\omega'$ are purely imaginary.
\end{prop}
\begin{proof}
Let a hyperelliptic curve \eqref{V22g1Eq}  have the form $-y^2 + \Lambda(x)=0$,
where by $\Lambda(x)$ the polynomial in $x$ is denoted.
If all branch points are real, then $\Lambda(x) > 0$ at $x\in (e_{2k-1},e_{2k})$, $k=1$, \ldots,~$g$.
Thus, the periods $\omega_k$, computed from the holomorphic differentials \eqref{K1DifsGen}
along $\mathfrak{a}_k$ cycles, are real, cf.\,\eqref{K12PerComp}. 
On the other hand, $\Lambda(x) < 0$ at $x\in (e_{2k},e_{2k+1})$, $k=1$, \ldots, $g$.
And so the periods  $\omega'_k$, computed from the holomorphic differentials \eqref{K1DifsGen}
along $\mathfrak{b}_k$ cycles, are purely imaginary, cf.\,\eqref{K12PerComp}.

Let a hyperelliptic curve \eqref{CurveGN} have the form $-w^2 + H(z)=0$, where
the leading term  $z^{2g+1}$ of the polynomial $H(z)$ has an arbitrary real coefficient $\rmb^2$.
By the transformation $z\mapsto x$, $w \mapsto \rmb y$
we find $(w/\rmb)^2 = \Lambda(z)$. Thus, the holomorphic differentials
\eqref{K1Difs} produce the same periods $\omega_k$, $\omega'_k$ as  \eqref{K1DifsGen}.

Finally, suppose, that the leading term of a genus $g$ hyperelliptic curve of the form $-y^2 + \Lambda(x)=0$
is $x^{2g+2}$, and so all $2g+2$ branch points are finite. In this case, we enumerate branch points 
by indices $i=0$, \ldots, $2g+1$, and $e_0$ serves as the base-point. The cut $(e_{2g+1}, \infty)$
is replaced by  $(e_{2g+1},\infty)\cup \{ \infty\} \cup (\infty,e_0)$. If the canonical cycles are defined
as on  fig.\,\ref{cyclesOdd}, then $\omega_k$ are real, and $\omega'_k$ are purely imaginary.
\end{proof}

\subsection{Singularities of $\wp$-functions}
Recall that all half-periods are described in terms of partitions 
$\I_\mFr \cup \J_\mFr$, $\mFr=0$, \ldots, $[(g+1)/2]$,
 of the set $\mathcal{S}$ of indices of branch points.
The multiplicity $\mFr$ shows the order of vanishing of $\theta[K]\big(v+\omega^{-1}\mathcal{A}(\I_\mFr)\big)$ at $v=0$,
and so the order of vanishing of  $\sigma \big(u+\mathcal{A}(\I_\mFr)\big)$.
Therefore, partitions with $\mFr = 0$ correspond to half-periods where the sigma function does not vanish. All other 
half-periods are zeroes of the sigma function, 
and so $\wp$-functions have singularities at $\mathcal{A}(\I_\mFr)$, $\mFr > 0$.

Let $|\I |$ denote the cardinality of a set $\I$. 
We drop $0$ from all sets, and calculate the cardinality omitting $0$.
Thus, $|\I_0 |=g$.

With a choice of cycles as on  fig.\,\ref{cyclesOdd}, we have
the following correspondence between sets $\I_{[g/2]}$ of cardinality $1$ and half-periods:
\begin{gather}\label{HPchar}
\begin{split}
&\{2k-1\} \sim \tfrac{1}{2} \omega'_{k} + \sum_{i=1}^{k-1} \tfrac{1}{2}  \omega_i,\quad
\{2k\} \sim  \tfrac{1}{2} \omega'_k + \sum_{i=1}^k \tfrac{1}{2}  \omega_i,\quad k=1,\dots, g, \\
&\{2g+1\} \sim  \sum_{i=1}^g \tfrac{1}{2}  \omega_i.
\end{split}
\end{gather}

Every half-period has the form $\Omega(I) + \Omega'(I')$,
where $\Omega(I)$ is generated from real half-periods $\tfrac{1}{2} \omega_k$, $k=1$, \ldots, $g$,
and $\Omega'(I')$ is generated from purely imaginary half-periods $\tfrac{1}{2}  \omega'_k$, $k=1$, \ldots, $g$.
Namely, 
\begin{equation}\label{ImPerShelf}
\Omega (I) = \sum_{k\in I} \tfrac{1}{2} \omega_k,\qquad 
\Omega' (I') = \sum_{k\in I'} \tfrac{1}{2} \omega'_k,
\end{equation}
and $I$, $I'$ are certain subsets of $\{1,\,2,\,\dots,\,g\}$. 
There exist $2^g$ such subsets. When two half-periods $\Omega (I_1) + \Omega' (I'_1)$ 
and $\Omega (I_2) + \Omega' (I'_2)$  are added, the resulting
subset $I$ is the union of $I_1$ and $I_2$ where indices occurring twice dropped,
and  $I'$ is obtained similarly from $I'_1$ and $I'_2$.

Let $\bar{\Jac} = \Complex^g$ be the vector space where the Jacobian variety $\Jac$ of a curve is embedded.
We split $\bar{\Jac} $ into the real part $\ReN \bar{\Jac}$, and the imaginary part $\ImN \bar{\Jac}$
such that $\bar{\Jac}  = \ReN \bar{\Jac} \oplus \ImN \bar{\Jac}$.
The real part $\ReN \bar{\Jac} \sim \Real^g$ is a span of real axes of $\bar{\Jac}$ over $\Real$.
And $\ImN \bar{\Jac}\sim \Real^g$ is a span of imaginary axes of $\bar{\Jac}$ over~$\Real$.
Consider $2^g$ subspaces $\mathfrak{U}_{\Im}(I) = \Omega(I) + \ImN \bar{\Jac}$,
 parallel to $\ImN \bar{\Jac}$, and  $2^g$  
subspaces $\mathfrak{U}_{\Re}(I') = \Omega'(I') + \ReN \bar{\Jac}$,
parallel to $\ReN \bar{\Jac}$.
We are interested in such subspaces where $\wp$-functions have no singularities.

\begin{prop}\label{P:URe}
Among $2^g$ subspaces $\mathfrak{U}_{\Re}(I') = \Omega'(I') + \ReN \bar{\Jac}$,
where $\Omega'(I')$ is defined by \eqref{ImPerShelf}, and $I'$ runs over all subsets of $\{1,\,2,\,\dots,\,g\}$,
there exists only one subspace which contains no zeros of the sigma function.
With a choice of cycles as on fig.\,\ref{cyclesOdd}, this subspace corresponds to
$\hat{I}' = \{1,2,\dots, g\}$, that is
 $$\Omega'(\hat{I}' ) = \sum_{k=1}^g \tfrac{1}{2}  \omega'_k.$$
 and so $u \in \mathfrak{U}_{\Re}(\hat{I}' ) $ has the form
\begin{gather*}
 u = s + \Omega'(\hat{I}' ), \quad s\in \Real^g.
\end{gather*}
\end{prop}
\begin{proof}
Due to \eqref{HPchar}, partitions of the form $\{\iota\} \cup \J_{[g/2]}$, $\iota = 1$, \ldots, $2g+1$, 
correspond to half-periods with the part $\Omega'(I')$ such that
$I'$ belongs to the collection $\mathfrak{I}'_1 = \big\{\emptyset$, $\{1\}$, $\{2\}$, \ldots, $\{g\}\big\}$.

If $g=1$, then $\{1\}$ is the required  $\hat{I}'$, and the line 
$\mathfrak{U}_\Re (\{1\}) = s + \tfrac{1}{2} \omega'_1$, $s \in \Real$,  
contains no zeros of the sigma function. Indeed, in the elliptic case ($g=1$), 
we have the following correspondence between characteristics of multiplicity $0$,
represented by partitions with $\I_0 =\{\iota\}$,  $\iota=1$, $2$, $3$, 
and half-periods:
$$
[\{1\}] \sim \tfrac{1}{2} \omega'_1, \quad 
[\{2\}] \sim  \tfrac{1}{2} \omega_1 + \tfrac{1}{2} \omega'_1,\quad 
[\{3\}] \sim \tfrac{1}{2} \omega_1.
$$
The sigma function does not vanish at these half-periods.
Within the fundamental domain, the only zero of the sigma function is located at $u=0$.

If $g>1$, then each $\I_0$ is obtained from $g$ sets $\{\iota\}$ with all indices different.
The corresponding $I'$ is obtained by taking the union of $g$ subsets from the collection $\mathfrak{I}'_1$,
and dropping indices which occur even number of times.
If at least two subsets in this union coincide, then the resulting $I'$ is a union
of $g-2$ or less number of subsets from $\mathfrak{I}'_1$. 
That means, that the subspace $\mathfrak{U}_{\Re}(I')$
contains zeros of the sigma function. Thus, the required $\hat{I}'$ is obtained
by the union of $g$ different and not empty subsets from the collection $\mathfrak{I}'_1$.
This implies $\hat{I}' = \{1,2, \dots, g\}$.
\end{proof}

\begin{prop}\label{P:UIm}
Among $2^g$ subspaces $\mathfrak{U}_{\Im}(I) = \Omega(I) + \ImN \bar{\Jac}$, 
where $\Omega(I)$ is defined by \eqref{ImPerShelf}, and $I$ runs over all subsets of $\{1,\,2,\,\dots,\,g\}$, 
there exists only one subspace which contains no zeros of the sigma function.
With a choice of cycles as on fig.\,\ref{cyclesOdd}, this subspace corresponds to
$\hat{I} = \{1,3,\dots, g\}$, if $g$ is odd, or $\hat{I} = \{2,4,\dots, g\}$, if $g$ is even, that is
 $$\Omega(\hat{I}) = \left\{
 \begin{array}{ll}
 \sum_{k=1}^{\kFr} \tfrac{1}{2}  \omega_{2k}, & g= 2\kFr\\
  \sum_{k=0}^{\kFr} \tfrac{1}{2}  \omega_{2k-1} & g= 2\kFr-1,
 \end{array}
 \right. $$
 and so $u \in \mathfrak{U}_{\Im}(\hat{I}) $ has the form
\begin{gather*}
 u = \imath s + \Omega(\hat{I}), \quad s \in \Real^g.
\end{gather*}
\end{prop}

\begin{proof} We use the same idea as in the proof of Proposition~\ref{P:URe}.
Due to \eqref{HPchar}, partitions of the form $\{\iota\} \cup \J_{[g/2]}$, $\iota = 1$, \ldots, $2g+1$, 
correspond to half-periods with the part $\Omega(I)$ such that
$I$ belongs to the collection $\mathfrak{I}_1 = \big\{\emptyset$, $\{1\}$, $\{1,2\}$, \ldots, $\{1,2,\ldots, g\}\big\}$.
The required $\hat{I}$ corresponds to $\I_0$ of cardinality $g$, and is obtained
by the union of $g$ different and not empty subsets from the collection $\mathfrak{I}_1$.
If $g$ is even, then $\hat{I}$ contains only even numbers between $1$ and $g$.
If $g$ is odd, then $\hat{I}$ contains only odd numbers between $1$ and $g$.
\end{proof}

%-----------------------------------------------
\subsection{Hyperelliptic addition law}
Below, we briefly recall the addition laws on hyperelliptic curves, formulated in \cite{bl2005}.

Let  $\Upsilon_n(u)$, $n=2$, \ldots, $g+2$, be $g$-component vector-functions of $u\in \Jac$.
We introduce a matrix-function $\Qp(u) = \big(\Upsilon_{2}(u)$, $\Upsilon_3(u)$, \ldots, 
$\Upsilon_{g+1}(u)\big)$ with entries $(\textsf{q}_{i,j}(u))$, and a vector-function 
$\bm{q}(u) = \Upsilon_{g+2}(u)$. Actually, on a curve of the form~\eqref{V22g1Eq}
\begin{subequations}
\begin{align}
&\Upsilon_2(u) \equiv \big(\textsf{q}_{i,1}(u)\big) = \big(\wp_{1,2i-1}(u)\big),\quad i=1,\dots, g,\\
&\Upsilon_3(u) \equiv \big(\textsf{q}_{i,2}(u)\big) = -\tfrac{1}{2} \big(\wp_{1,1,2i-1}(u)\big),\\
&\Upsilon_{k+1}(u) \equiv (\textsf{q}_{i,k}(u)) = 
 \textsf{q}_{1,k-2}(u) \Upsilon_{2}(u) +  \Upsilon^\circ_{k-1}(u),\quad k=3,\dots,g+2, \label{UpsilonRule}\\
 &\text{where } \Upsilon^\circ_{k-1}(u) = \big( \textsf{q}_{2,k-2}(u),  \textsf{q}_{3,k-2}(u),\dots,
  \textsf{q}_{g,k-2}(u), 0\big)^t.
\end{align}
\end{subequations}
Let $\bar{\nu} = (\nu_{g+2}, \nu_{g+4},  \dots, \nu_{3g} )^t$, 
and $\nu = (\nu_g, \dots, \nu_2, \nu_1)^t$. 
Let $u_\text{I}$, $u_\text{II}$, $u_\text{III}\in \Jac$, subject to $u_\text{I}+u_\text{II}+u_\text{III} = 0$. 
Then the system
\begin{gather*}
\begin{pmatrix}
1_g & \Qp(u_\text{I}) \\
1_g & \Qp(u_\text{II}) \\
1_g & \Qp(u_\text{III}) 
\end{pmatrix} 
\begin{pmatrix}  \bar{\nu} \\ \nu \end{pmatrix} = -
\begin{pmatrix}  \bm{q}(u_\text{I}) \\ \bm{q}(u_\text{II}) \\ \bm{q}(u_\text{III}) \end{pmatrix} 
\end{gather*}
defines the addition law.

Let a $g\times g$ matrix $\Pp \equiv (\textsf{p}_{i,j})_{i,j=1}^g$ be defined as follows 
$$\Pp = \big(\Upsilon_{2}(u_\text{I})+\Upsilon_{2}(u_\text{II}), \Upsilon_4(u_\text{I})+\Upsilon_4(u_\text{II}), \dots, 
\Upsilon_{2g}(u_\text{I})+\Upsilon_{2g}(u_\text{II})\big),$$
where $\Upsilon_k$, $k>g+2$, are computed by the  rule \eqref{UpsilonRule}. 
Let a vector $\Pi = (\Pi_0$, $\Pi_1$, \ldots, $\Pi_g)$ be defined as follows
\begin{subequations}\label{WPijCompnt}
\begin{gather}
\Pi_0 = 1,\qquad
\Pi_{k} = \frac{1}{k} \sum_{i=1}^{k} \Pi_{k-i} \sum_{j=1}^i  \textsf{p}_{i-j+1,j},\quad
k=1,\dots, g,
\end{gather}
and a vector $\mathrm{N} = (\mathrm{N}_0$, $\mathrm{N}_1$, \ldots, $\mathrm{N}_g)$, $\mathrm{N}_0=-1$,
depending on the parity of $g$, be
\begin{gather}\label{HkCompnt}
\begin{split}
&g=2\kFr-1:\qquad
\mathrm{N}_k = \sum_{i=1}^k \nu_{2k-2i+1} \nu_{2i -1} - 
\sum_{j=0}^k \lambda_{2k-2j} \sum_{i=0}^j \nu_{2j-2i} \nu_{2i},\\
&g=2\kFr:\qquad\quad\  \ 
\mathrm{N}_k = - \sum_{i=0}^k \nu_{2k-2i} \nu_{2i} + 
\sum_{j=1}^k \lambda_{2k-2j} \sum_{i=1}^j \nu_{2j-2i+1} \nu_{2i-1},
\end{split}
\end{gather}
\end{subequations}
where $\nu_0=1$, and $\lambda_0=1$.
Then addition formulas for $\wp_{1,2i-1}$ are given by
\begin{gather}\label{WPijAddLaw}
\wp_{1,2i-1}(u_\text{III}) = \sum_{j=0}^i \mathrm{N}_{i-j} \Pi_{j}.
\end{gather}
They work in an arbitrary genus $g$, and $\nu_k$ is supposed to be zero if 
there is no such entry of $\bar{\nu}$, $\nu$ in this genus.
In particular, 
\begin{gather}\label{WP11AddLaw}
\begin{split}
&g=2\kFr-1:\qquad
\wp_{1,1}(u_\text{III}) = \nu_1^2 - 2 \nu_2 - \lambda_2 - \wp_{1,1}(u_\text{I}) - \wp_{1,1}(u_\text{II}),\\
&g=2\kFr:\qquad\quad\  \ 
\wp_{1,1}(u_\text{III}) = \nu_1^2 - 2 \nu_2 - \wp_{1,1}(u_\text{I}) - \wp_{1,1}(u_\text{II}),
\end{split}
\end{gather}
where $\nu_2=0$ in the case of genus $1$. 

Expressions for $\nu_k$ in terms of $\wp_{1,2i-1}(u_\text{I})$, 
$\wp_{1,2i-1}(u_\text{II})$, $\wp_{1,1,2i-1}(u_\text{I})$, $\wp_{1,1,2i-1}(u_\text{II})$ are obtained from the system
\begin{gather}\label{EtaEqs}
\begin{pmatrix}
1_g & \Qp(u_\text{I}) \\
1_g & \Qp(u_\text{II})
\end{pmatrix} 
\begin{pmatrix}  \bar{\nu} \\ \nu \end{pmatrix} = 
- \begin{pmatrix}  \bm{q}(u_\text{I}) \\ \bm{q}(u_\text{II}) \end{pmatrix}.
\end{gather}
Expressions for $\wp_{1,1,2i-1}(u_\text{III})$ in terms of $\wp_{1,2i-1}(u_\text{III})$ and $\nu_k$
are obtained from 
\begin{gather}\label{WPijkAddLaw}
\bar{\nu} + \Qp(u_\text{III}) \nu = - \bm{q}(u_\text{III}),
\end{gather}
and the substitution \eqref{WP11AddLaw} turns these into addition formulas.

\begin{rem}
The addition law is obtained from the entire 
rational function $\mathcal{R}_{3g}(x,y) = y \nu_y(x) + \nu_x(x)$ of weight $3g$,
where $\nu_y(x) = \sum_{i=1}^{[(g-1)/2]} \nu_{g-1-2i} x^i$, and $\nu_x(x) = \sum_{i=0}^{[3g/2]} \nu_{3g-2i} x^i$.
With the help of the solution \eqref{EnC22g1} of the Jacobi inversion problem, 
$y$ is eliminated from $\mathcal{R}_{3g}$, and the degree of $x$ is reduced to $g-1$.
Pre-images of $u_\text{I}$, $u_\text{II}$, $u_\text{III}$ 
are supposed to be zeros of $\mathcal{R}_{3g}$. Thus, coefficients of $\mathcal{R}_{3g}$ 
reduced to a degree $g-1$ polynomial in $x$ produce $3g$ equations \eqref{EtaEqs}, \eqref{WPijkAddLaw}. 
The expressions \eqref{WPijAddLaw} are derived from the equality
$$\nu_y^2 f(x,- \nu_x / \nu_y) = 
\mathcal{R}_{2g}(x,y;u_\text{I}) \mathcal{R}_{2g}(x,y;u_\text{II}) \mathcal{R}_{2g}(x,y;u_\text{III}),$$
which reflects the fact that $u_\text{I}$, $u_\text{II}$, $u_\text{III}$  form the divisor of zeros of $\mathcal{R}_{3g}$.
\end{rem}

%-----------------------------------------------
\subsection{Real-valued $\wp$-functions}
Through the help of the addition law one can obtain expressions for
$\wp_{1,2i-1}(u+ \imath \upsilon)$ and $\wp_{1,1,2i-1}(u+ \imath \upsilon)$
 in terms of $\wp_{1,2i-1}(u)$, $\wp_{1,2i-1}(\imath \upsilon)$, 
$\wp_{1,1,2i-1}(u)$, $\wp_{1,1,2i-1}( \imath \upsilon)$,
$i=1$, \ldots, $g$.
The mentioned $2g$ $\wp$-functions serve as generators in the differential field
of all multiply periodic functions on a hyperelliptic curve of genus $g$, see \cite{bl2008}. 
These $2g$ functions
arise in the solution \eqref{EnC22g1} of the Jacobi inversion problem.

Below, we prove that $\wp_{i,j}$ on  subspaces $\mathfrak{U}_\Re (I')$, and $\mathfrak{U}_\Im (I)$
are real-valued. 

\begin{prop}\label{P:WPiv0}
Let  $\upsilon \in \Real^g$ be fixed, and  $\wp_{1,1,2i-1}(\imath \upsilon)=0$,
that is $\imath \upsilon$ is a  half-period. 
Then for all $u\in \Real^g$,
functions $\wp_{1,2i-1}(u+ \imath \upsilon)$
and $\wp_{1,1,2i-1}(u+ \imath \upsilon)$, $i=1$, \ldots, $g$, are real-valued.
\end{prop}
\begin{proof}
The fact that $\wp_{1,1,2i-1}$ vanishes at all half-periods 
on a hyperelliptic curve follows immediately  from \eqref{R2g1}.

Next, recall that $\wp_{i,j}$ are even functions, and so
$\wp_{i,j}(\imath \upsilon) = \widehat{\wp}_{i,j}(\upsilon)$ are real-valued. 
On the other hand, $\wp_{i,j,k}$ are odd, and so  
$\wp_{i,j,k}(\imath \upsilon)= \imath \widehat{\wp}_{i,j,k}(\upsilon)$ are purely imaginary.
Since $\imath$ emerges only from
$\wp_{1,1,2i-1}(\imath \upsilon)$, which vanishes  according to the condition,
all functions $\wp_{1,2i-1}(u+\imath \upsilon)$
and $\wp_{1,1,2i-1}(u+ \imath \upsilon)$ are real-valued.
\end{proof}

\begin{prop}\label{P:WPu0}
Let  $u \in \Real^g$ be fixed, and  $\wp_{1,1,2i-1}(u)=0$,
that is $u$ is a  half-period.  Then for all $\upsilon \in \Real^g$, functions
 $\wp_{1,2i-1}(u+ \imath \upsilon)$, $i=1$, \ldots, $g$, 
are real-valued,
and $\wp_{1,1,2i-1}(u+ \imath \upsilon)$, $i=1$, \ldots, $g$, are purely imaginary.
\end{prop}

\begin{proof}
From the structure of $\Qp$ and $\bm{q}$ in \eqref{EtaEqs} we see the following.

If $g = 2\kFr-1$, then $\Qp$ contains $\kFr-1$ columns linear in odd functions $\wp_{1,1,2i-1}$,
and $\bm{q}$ is linear in $\wp_{1,1,2i-1}$. These columns of $\Qp(\imath \upsilon)$ 
and $\bm{q}(\imath \upsilon)$ have purely imaginary entries.
Due to $\wp_{1,1,2i-1}(u)=0$, the corresponding columns of $\Qp(u)$ vanish, as well as $\bm{q}(u)$.
The system \eqref{EtaEqs} with $u_{\text{I}}=u$ and $u_{\text{II}}=\imath \upsilon$ is solved by Cramer's rule. 
Thus, all entries of $\bar{\nu}$, which are indexed
by odd numbers, have purely imaginary values, as well as the entries $\nu_{k}$ of  $\nu$ with odd $k$.
The entries $\nu_{k}$ with even $k$ are real-valued. 

If $g = 2\kFr$, then $\Qp$ contains $\kFr$ columns linear in odd functions $\wp_{1,1,2i-1}$,
and $\bm{q}$ is expressed in terms of even functions $\wp_{1,2i-1}$ only.
Thus, all entries of $\bar{\nu}$, which are indexed
by even numbers, are real-valued, as well as the entries $\nu_{k}$ of $\nu$ with even~$k$.
The entries $\nu_{k}$ with odd $k$ are purely imaginary. 

Therefore,  $\nu_{k}$ with odd $k$ are purely imaginary, and
with even $k$ are real-valued, in any genus. 

Next, we look at the expressions for $\textsf{N}_k$
given by \eqref{HkCompnt}. Evidently, all $\textsf{N}_k$ are real-valued,
since $\nu_k$ with odd $k$ appear in quadratic terms $\nu_{2j-2i+1} \nu_{2i-1}$,
which are real-valued. Therefore, expressions 
for $\wp_{1,2i-1}(u+ \imath \upsilon)$ given by \eqref{WPijAddLaw} 
with $u+ \imath \upsilon = - u_{\text{III}}$ are real-valued.

Expressions for $\wp_{1,1,2i-1}(u+ \imath \upsilon)$ are obtained from 
\eqref{WPijkAddLaw} with $u+ \imath \upsilon = - u_{\text{III}}$. If $g = 2\kFr-1$, then $\nu_k$ with odd $k$
are multiplied by real-valued columns expressed in terms of even functions $\wp_{1,2i-1}(u_{\text{III}})$ only.
Since $\nu_k$ with odd $k$ are purely imaginary, the expression 
$\bar{\nu} + \sum_{i=1}^\kFr \nu_{2i-1} \Upsilon_{2\kFr-2i+2}(u_{\text{III}})$
has purely imaginary entries. On the other hand, $\nu_k$ with even $k$ are real-valued,
and they arise as multiples of the columns linear in odd functions $\wp_{1,1,2i-1}(u_{\text{III}})$.
Thus, $\wp_{1,1,2i-1}(u_{\text{III}})$ are purely imaginary. Similar considerations in the case of
$g = 2\kFr$ bring to the same conclusion.
\end{proof}

 Proposition~\ref{P:ReImPeriods}
shows that all $\Omega(I)$ are real, and all $\Omega'(I')$ are purely imaginary,
if all branch points of the curve are real. Therefore,
according to Proposition~\ref{P:WPiv0}, on 
$\mathfrak{U}_\Re (I') = \Omega'(I') + \ReN \bar{\Jac}$ functions
$\wp_{1,2i-1}$ and $\wp_{1,1,2i-1}$, $i=1, \ldots$, $g$, are real-valued.
According to Proposition~\ref{P:WPu0}, on $\mathfrak{U}_\Im (I) = \Omega(I) + \ImN \bar{\Jac}$
functions $\wp_{1,2i-1}$, $i=1, \ldots$, $g$, are real-valued, and
$\wp_{1,1,2i-1}$, $i=1, \ldots$, $g$, take only purely imaginary values.

In the finite-gap solution  \eqref{KdVSolRealCond} we choose 
$$\bm{C} = \Omega'(\hat{I}') + \mathbf{c} = \sum_{k=1}^g \tfrac{1}{2}  \omega'_k +
(\mathrm{c}_1, \mathrm{c}_3, \mathrm{c}_5, \dots, \mathrm{c}_{2N-1})^t,$$ 
with $\hat{I}'$ defined in Proposition~\ref{P:URe}, and
arbitrary real $\mathrm{c}_{2n-1}$, $n=1$, \ldots, $N$.

\begin{rem}\label{R:CK}
Note that the vector of Riemann constants $K$ is expressed in terms of
 $\Omega'(\hat{I}')$ and $\Omega(\hat{I})$ defined in Propositions~\ref{P:URe} and \ref{P:UIm}, namely   
 $$K = \omega^{-1} \big(\Omega(\hat{I}) + \Omega'(\hat{I}') \big).$$
 Therefore, it is convenient to assign $\bm{C} = \omega K + \mathbf{c}$.
\end{rem}

\begin{rem}\label{r:Comparison}
Comparing the KdV solution in the form of \eqref{KdVSolRealCond} with 
\eqref{KdVSolLogTheta}, we find the following correspondence:
\begin{gather*}
\bm{U} = -\tfrac{1}{2} \omega^{-1}_1,\quad \bm{W} = -\tfrac{1}{2} \omega^{-1}_2,\quad
\bm{D} =  \omega^{-1} \bm{C},\quad c = -\tfrac{1}{4} \varkappa_{1,1},
\end{gather*}
where $\omega^{-1}_k$ denotes the $k$-th column of the matrix $\omega^{-1}$,
and $\varkappa = \eta \omega^{-1}$.
\end{rem}

%---------------------------------------
\subsection{Summary}
In $\bar{\Jac} = \Complex^g$, where the Jacobian variety of the spectral curve
is embedded, among $2^g$ affine subspaces $\mathfrak{U}_{\Re}$ parallel to the real axes subspace $\ReN \bar{\Jac}$ 
and obtained by half-period translations,
there exists one subspace where $\wp$-functions are bounded.
With the choice of cycles as on fig.\,\ref{cyclesOdd},
this subspace is $\sum_{k=1}^g \tfrac{1}{2}  \omega'_k + \ReN \bar{\Jac}$.
And among $2^g$ affine subspaces $\mathfrak{U}_{\Im}$ parallel to the imaginary axes subspace $\ImN \bar{\Jac}$ 
and obtained by half-period translations,
 there exists one subspace $\sum_{i=0}^{[(g-1)/2]} \tfrac{1}{2}  \omega_{g-2i} + \ImN \bar{\Jac}$, 
 where  $\wp$-functions are bounded.  Then, by means of the addition law on the curve,
 it was proven that functions $\wp_{1,2i-1}$ on all $\mathfrak{U}_{\Re}$ and $\mathfrak{U}_{\Im}$ subspaces are real-valued,
 and  functions $\wp_{1,1,2i-1}$ are real-valued on subspaces $\mathfrak{U}_{\Re}$,
 and purely imaginary-valued on subspaces $\mathfrak{U}_{\Im}$.
Thus, if all finite branch points of the spectral curve are real, 
as required for quasi-periodic solutions of the KdV equation, 
then $\sum_{k=1}^g \tfrac{1}{2}  \omega'_k + \ReN \bar{\Jac}$ serves as the domain
of the bounded real-valued $\wp_{1,1}$-function.

%======================================
\section{Non-linear waves}\label{s:NLW}
In this section we present effective numerical computation of quasi-periodic 
finite-gap solutions of the KdV equation.

An analytical approach to computation of $\wp$-functions is used, for more details see \cite{BerCalc24}.
Once a curve is chosen, periods  of the first kind \eqref{K12PerComp}, and the corresponding 
periods  of the second kind  are computed. Then $\wp$-functions are calculated by the formulas
\eqref{WPdefComp}. All computations and graphical representation are performed in Wolfram Mathematica~12.
The function \texttt{NIntegrate} is used for numerical integration of periods,
and \texttt{NSolve} for computing branch points in genera $2$ and $3$.
The default precision is applied, which is sufficient for  graphical representation.
The theta function is approximated by a partial sum of \eqref{ThetaDef}, 
$|n_i| \leqslant 5$  is sufficient.

%------------------------------------------------------------------
\subsection{Numerical computation}
First, the curve \eqref{CurveGN} is transformed to the canonical form \eqref{V22g1Eq}
as follows
\begin{multline}\label{CurveGNCanon}
 0 = \frac{1}{\rmb^2} F(z,w) \equiv - \Big( \frac{w}{\rmb}\Big)^2 + z^{2N+1} 
 + \frac{c_{2N}}{\rmb^2 } z^{2N} + \cdots + \frac{c_N}{\rmb^2 } z^N \\
 + \frac{h_{N-1}}{\rmb^2 } z^{N-1} + \dots + \frac{h_1}{\rmb^2 } z + \frac{h_0}{\rmb^2 },
\end{multline}
where $z=x$, $w/ \rmb = y$, and so
 $\lambda_{4g+2-2n} = h_n / \rmb^2$.
The basis differentials of the first and second kind \eqref{K12Difs}
correspond to this canonical form. In this setup the sigma function 
is defined uniquely from heat equations, and can be obtained in the form of a series
in coordinates $u$ of the Jacobian variety, and parameters $\lambda_{2i+2}$ of the curve.
However, such a series converges slowly, which is not acceptable for computational purposes.
The situation with $\wp$-functions is even worse.

So the formula \eqref{SigmaThetaRel} is employed to compute the sigma function.
Correspondingly, we compute $\wp$-functions by
\begin{gather}\label{WPdefComp}
\begin{split}
&\wp_{i,j}(u)  = \varkappa_{i,j} - \frac{\partial^2}{\partial u_i \partial u_j } \log \theta[K](\omega^{-1} u; \omega^{-1} \omega'),\\
& \wp_{i,j,k}(u)  =- \frac{\partial^3}{\partial u_i \partial u_j \partial u_k} \log \theta[K](\omega^{-1} u; \omega^{-1} \omega').
\end{split}
\end{gather}

The period matrices $\omega$,  $\omega'$,  $\eta$ are computed from 
the differentials \eqref{K1Difs}, \eqref{K2Difs} along the canonical cycles, see
fig.\,\ref{cyclesOdd}. By $\varkappa_{i,j}$ entries of the symmetric matrix $\varkappa = \eta \omega^{-1}$
are denoted. Actually, columns of the matrices $\omega$,  $\omega'$,  $\eta$
are computed as follows
\begin{gather}\label{K12PerComp}
\begin{split}
  &\omega_k = 2\int_{e_{2k-1}}^{e_{2k}} \rmd u,\qquad\quad \eta_k = 2\int_{e_{2k-1}}^{e_{2k}} \rmd r,\\
  &\omega'_k = - 2 \sum_{i=1}^k \int_{e_{2i-2}}^{e_{2i-1}} \rmd u  
  = 2 \sum_{i=k}^g  \int_{e_{2i}}^{e_{2i+1}} \rmd u.
  \end{split}
\end{gather}
The latter equality holds due to
\begin{gather*}
\sum_{i=0}^g  \int_{e_{2i}}^{e_{2i+1}} \rmd u = 0.
\end{gather*}

Quasi-periodic solutions of the KdV equation arise when all finite branch points $\{e_i\}_{i=1}^{2N+1}$ of the spectral curve \eqref{CurveGNCanon} are real. In what follows, we assign $c_{2N}=0$, as we have in the KdV equation.
On the other hand, by the transformation $z \mapsto z - (2g+1)^{-1} c_{2N} \rmb^{-2}$ 
the term $z^{2N}$ is eliminated from \eqref{CurveGNCanon}.

Below, we illustrate the proposed approach to
computing quasi-periodic solutions of the KdV equation
 in genera 2 and 3, and compare the obtained result with
the known one in genus $1$.

%------------------------------------------------------------------
\subsection{Genus $1$}
The hamiltonian system of the KdV equation in $\M_1^\circ$ possesses
 the spectral curve 
 \begin{equation}\label{G1HamCanon} 
-w^2 + \rmb^2 z^3 + c_2 z^2 + c_1 z + h_0 = 0.
 \end{equation}
 In this particular case, $c_2$ is kept non-vanishing for a while.
Looking for real-valued solutions, 
we suppose that all  parameters of the curve: $\rmb$, $c_2$, $c_1$, $h_0$ are real,
and  the spectral curve has three real branch points  $e_1 < e_2 < e_3$.

According to \eqref{JIPb}, uniformization of \eqref{G1HamCanon} is given by
\begin{gather}\label{UniformG1}
z = \wp_{1,1}(u),\qquad w = - \tfrac{1}{2} \rmb \wp_{1,1,1}(u),
\end{gather}
where $u = u_1$ is complex, in general. Note, that $\wp$-function in \eqref{UniformG1} 
corresponds to the curve \eqref{G1HamCanon}, and relates
to the Weierstrass function $\wp(u;g_2,g_3)$ as follows
\begin{gather*}
\begin{split}
\wp_{1,1} (u) &= \wp(u;g_2,g_3) - \tfrac{1}{3} c_2 \rmb^{-2}, \\ 
&g_2 = - 4 \big(c_1- \tfrac{1}{3}  c_2^2 \rmb^{-2}\big)\rmb^{-2},\\
&g_3 = - 4 \big(h_0 - \tfrac{1}{3} c_1 c_2 \rmb^{-2} + \tfrac{2}{27} c_2^3 \rmb^{-4}\big)\rmb^{-2}.
\end{split}
\end{gather*}

Applying the reality conditions to \eqref{KdVSolRealCond},
we find bounded real-valued solutions of  \eqref{KdVEqDegen}, 
and so the stationary KdV equation ($N=1$):
\begin{align}\label{betaKdVN1}
&\beta (\x) = -\rmb \wp_{1,1} (- \rmb \x + \tfrac{1}{2}\omega').
\end{align}

Let $\M_1^\circ$ be fixed by $4\rmb^2 =1$,  $c_2=0$, $c_1 = -588$.
Let the hamiltonian $h_0$ attains three values with different mutual positions of branch points:

\begin{tabular}{ll}
$h_0 = -10894$ & $e_1 = 13 - 3\sqrt{205} \approx -29.95$, \\
& $e_2 = -26$, \\
& $e_3 = 13 + 3\sqrt{205}\approx 55.95$; \\
$h_0 = 0$ & $e_1 = - 28 \sqrt{3} \approx -49.50$, \\
& $e_2 = 0$, \\
& $e_3 = 28 \sqrt{3} \approx 49.50$;\\
$h_0 = 10894$ & $e_1 = -13 - 3\sqrt{205} \approx -55.95$, \\
& $e_2 = 26 $, \\
& $e_3 =  -13 + 3\sqrt{205} \approx 29.95$. 
\end{tabular}\\
The corresponding fundamental domains and shapes of \eqref{betaKdVN1} with $\rmb = 1/2$ 
are shown on fig.~\ref{WaveG1}.
\begin{figure}[h]
\caption{The case $N=1$.} \label{WaveG1}
2a. Fundamental domains.
$\phantom{mmmm}$ 2b. Cnoidal waves.$\displaystyle\phantom{\dfrac{A}{A}}\qquad\qquad\qquad\qquad\qquad$
\parbox[b]{0.3\textwidth}{
\includegraphics[width=0.28\textwidth]{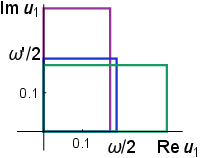}\\
$\quad \vphantom{\displaystyle \int} $ } 
\parbox[b]{0.47\textwidth}{
\includegraphics[width=0.43\textwidth]{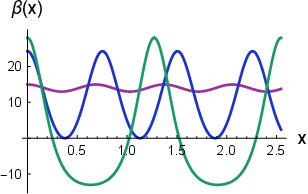}}
\parbox[b]{0.2\textwidth}{
\includegraphics[width=0.19\textwidth]{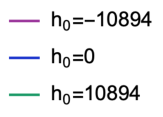} \\
$\quad \vphantom{\displaystyle \sum_y^T}$ }
\end{figure}
Note, if $\rmb$ is positive, than $\beta$ has maxima $- \rmb e_1$ at $\x = n \omega/\rmb$, $n\in\Integer$, 
and minima $- \rmb e_2$ at $\x = (n + \frac{1}{2})\omega/\rmb$, $n\in\Integer$. It follows from the fact that 
$\wp(\tfrac{1}{2}\omega') = e_1$, $\wp(\tfrac{1}{2}\omega + \tfrac{1}{2}\omega') = e_2$.

In terms of the Jacobi elliptic functions, we obtain the cnoidal wave, found by 
Korteweg and de Vries \cite{KdV}:
\begin{align*}%\label{CNwave}
-\tfrac{1}{2}\wp(-\tfrac{1}{2} \x + \tfrac{1}{2}\omega') &+ \tfrac{1}{6} (e_1 + e_2 + e_3) \\
& =  -\tfrac{1}{2} \bigg(e_2 + (e_3-e_1) 
\frac{\dn \big(\imath K' \mp \tfrac{1}{2} \x \sqrt{e_3-e_1}; k\big)^2}
{\sn \big(\imath K' \mp \tfrac{1}{2} \x \sqrt{e_3-e_1}; k\big)^2} \bigg) \notag  \\
&=  -\tfrac{1}{2}\Big(e_2 -  (e_2-e_1) \cn \big(\tfrac{1}{2} \x \sqrt{e_3-e_1}; k\big)^2 \Big), \notag
\end{align*}
where $k^2 = (e_2-e_1)/(e_3-e_1)$, and $c_2 \rmb^{-2} = - (e_1 + e_2 + e_3)$.

%------------------------------------------------------------------
\subsection{Genus 2}
The hamiltonian system of the KdV equation in $\M_2^\circ$ possesses
 the spectral curve 
 \begin{equation}\label{G2HamCanon} 
-w^2 + \rmb^2 z^5 + c_3 z^3 + c_2 z^2 + h_1 z + h_0 = 0.
 \end{equation}
Requiring real-valued solutions, 
we suppose that all  parameters $\rmb$, $c_3$, $c_2$, $h_1$, $h_0$ are real,
and  the spectral curve has five real branch points  $e_1 < e_2 < e_3 < e_4 < e_5$.
According to \eqref{JIPb}, uniformization of \eqref{G2HamCanon} is given by
\begin{gather}\label{UniformG2}
z^2 = z \wp_{1,1}(u) + \wp_{1,3}(u), \qquad 
w = - \tfrac{1}{2} \rmb \big(z \wp_{1,1,1}(u) + \wp_{1,1,3}(u)\big),
\end{gather}
where $u = (u_1,u_3)$, and the both components are complex, in general.

Bounded real-valued solutions in the case of $N=2$ are
\begin{align}\label{betaKdVN2}
&\beta (\x,\rmt) = -\rmb \wp_{1,1} (- \rmb \x + C_1, - \rmb \rmt + C_3),\\
\intertext{and by the reality conditions}
 \notag
\begin{split}
&C_1 = \tfrac{1}{2}\omega'_{1,1} + \tfrac{1}{2}\omega'_{1,2}, 
\qquad C_3 = \tfrac{1}{2}\omega'_{3,1} + \tfrac{1}{2}\omega'_{3,2},
\end{split}
\end{align}
where $(\omega'_{1,k},\omega'_{3,k})^t = \omega'_k$. First indices of the entries $\omega'_{2i-1,k}$
correspond to the labels of holomorphic differentials $\rmd u_{2i-1}$, $i=1$, $2$,
and $\omega'_1$, $\omega'_2$ are two columns of the period matrix $\omega'$.
On the subspace $-(\rmb \x,\rmb \rmt)^t +  \tfrac{1}{2}\omega'_{1} +  \tfrac{1}{2}\omega'_{2}$ 
the function $\wp_{1,1}$ is real-valued and bounded, with the critical values:
\begin{gather}\label{WP11CritValN2}
\begin{split}
&\wp_{1,1}(\tfrac{1}{2} \omega'_1 + \tfrac{1}{2} \omega'_2) = e_2  + e_3, \\
&\wp_{1,1}(\tfrac{1}{2} \omega_1 + \tfrac{1}{2} \omega'_1 + \tfrac{1}{2} \omega'_2) = e_1 + e_3,\\
&\wp_{1,1}(\tfrac{1}{2} \omega_2 + \tfrac{1}{2} \omega'_1 + \tfrac{1}{2} \omega'_2) = e_2 + e_4,\\
&\wp_{1,1}(\tfrac{1}{2} \omega_1 + \tfrac{1}{2} \omega_2 + \tfrac{1}{2} \omega'_1 + \tfrac{1}{2} \omega'_2) = e_1 + e_4.
\end{split}
\end{gather}

Let $\M_2^\circ$ be fixed by $4\rmb^2 =1$,  $c_4=0$, $c_3 = -15$,  $c_2=-20$.
Let the hamiltonians $h_1$, $h_0$ attains five values with different mutual positions of branch points:
\begin{tabular}{lll}
$h_1 = -3$, $h_0 = 1$ & $e_1 \approx -6.99$,  $e_2 \approx -1.13$, 
 $e_3 \approx -0.39$  $e_4 \approx 0.15$,  $e_5 \approx 8.35$; \\
& \parbox[]{0.62\paperwidth}{ 
\includegraphics[width=0.6\textwidth]{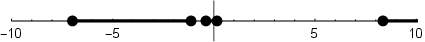} } 
 \end{tabular} \\
\begin{tabular}{lll}
$h_1 = 100$, $h_0 = -56$ & $e_1 \approx -5.69$,  $e_2 \approx -4.55$, 
 $e_3 \approx 0.72$  $e_4 \approx 1.50$,  $e_5 \approx 8.01$;\\
 & \parbox[]{0.62\paperwidth}{
 \includegraphics[width=0.6\textwidth]{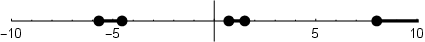} } 
 \end{tabular} \\
 \begin{tabular}{lll}
$h_1 = 100$, $h_0 = 81$ & $e_1 \approx -6.15$,  $e_2 \approx -3.52$, 
 $e_3 \approx -0.76$  $e_4 \approx 2.48$,  $e_5 \approx 7.94$;\\
 & \parbox[]{0.62\paperwidth}{
 \includegraphics[width=0.6\textwidth]{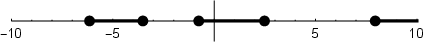} } \\
$h_1 = 100$, $h_0 = 156$ & $e_1 \approx -6.30$,  $e_2 \approx -2.71$, 
 $e_3 \approx -1.67$  $e_4 \approx 2.78$,  $e_5 \approx 7.90$;\\
 & \parbox[]{0.62\paperwidth}{
 \includegraphics[width=0.6\textwidth]{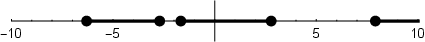} } \\
$h_1 = 210$, $h_0 = 460$ & $e_1 \approx -5.05$,  $e_2 \approx -3.86$, 
 $e_3 \approx -2.83$  $e_4 \approx 4.76$,  $e_5 \approx 6.99$.\\
 & \parbox[]{0.62\paperwidth}{
 \includegraphics[width=0.6\textwidth]{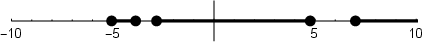} } 
\end{tabular}

The corresponding fundamental domains 
are shown in two projections: on $u_1$, and $u_3$, see fig.~\ref{FundDomG2}.
One can see a parallelogram spanned by $\frac{1}{2}\omega_{2i-1,1}$, $\frac{1}{2}\omega'_{2i-1,1}$
drawn with a dashed line,  a parallelogram spanned by $\frac{1}{2}\omega_{2i-1,2}$, $\frac{1}{2}\omega'_{2i-1,2}$
drawn with a dotted line, and a parallelogram spanned by $\frac{1}{2}(\omega_{2i-1,1} + \omega_{2i-1,2})$, 
$\frac{1}{2}(\omega'_{2i-1,1} + \omega'_{2i-1,2})$ drawn with a solid line.
\begin{figure}[h]
\caption{$N=2$.  Fundamental domains.} \label{FundDomG2}
3a. Projection on  $u_1$   $\phantom{mmmmmmmmmm}$  3b. Projection on  $u_3$
\parbox[b]{0.41\textwidth}{
\includegraphics[width=0.4\textwidth]{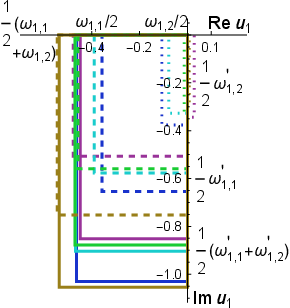} }
$\quad$
\parbox[b]{0.53\textwidth}{
\includegraphics[width=0.51\textwidth]{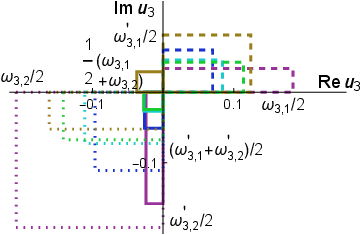}\\
$\quad  \vphantom{\displaystyle s^t}$\\
$\quad$}\\
\parbox[b]{0.24\textwidth}{
\includegraphics[width=0.23\textwidth]{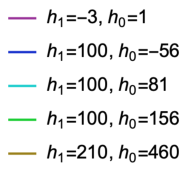}}$\qquad \qquad$
\parbox[b]{0.4\textwidth}{\includegraphics[width=0.4\textwidth]{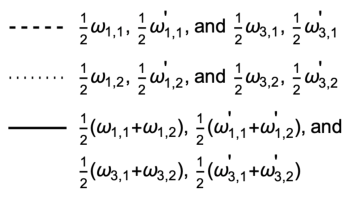} } 
\end{figure}

The corresponding shapes of  $\beta$, according to \eqref{betaKdVN2} with $\rmb = 1/2$, are presented 
on fig.~\ref{xtEvolG2} and \ref{CritValG2}. On fig.~\ref{xtEvolG2} the reader finds shapes of
$\beta(\x,0)$, $\x \in [0,-20 (\omega_{1,1}+\omega_{1,2})]$ (left column), and
$\beta(0,\rmt)$, $\rmt\in [0,20]$ (right column). On fig.~\ref{CritValG2} shapes of 
$\beta(\x,\rmt)$, $\x \in [0,-8 (\omega_{1,1}+\omega_{1,2})]$, $\rmt\in [0,-4 \omega_{3,2}]$
are presented for the chosen values of $h_1$ and $h_0$.
Dots on fig.~\ref{CritValG2} indicate critical values \eqref{WP11CritValN2}.  
\begin{figure}[h]
\caption{$N=2$. Quasi-periodic waves $\beta(\x,0)$ (left) and $\beta(0,\rmt)$ (right).} \label{xtEvolG2}
The cases $h_1 = -3$, $h_1 = 100$, $h_1 = 210$, from the top to the bottom.\\
\includegraphics[width=0.43\textwidth]{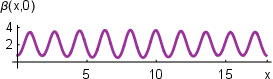}\ \ \ 
\includegraphics[width=0.53\textwidth]{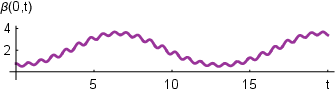} \\
\includegraphics[width=0.43\textwidth]{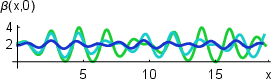}\ \ \ 
\includegraphics[width=0.53\textwidth]{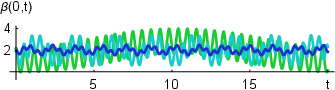} \\
\includegraphics[width=0.43\textwidth]{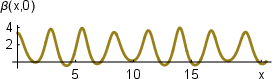}\ \ \ 
\includegraphics[width=0.53\textwidth]{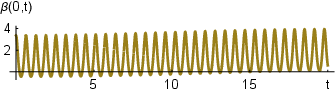} 
\end{figure}
\begin{figure}[h]
\caption{$N=2$. Quasi-periodic waves $\beta(\x,\rmt)$.} \label{CritValG2}
5a. $h_1=-3$, $h_0=1$.   $\phantom{mmmmmmmmmm}$  5b. $h_1=100$, $h_0=-56$. \\
\includegraphics[width=0.49\textwidth]{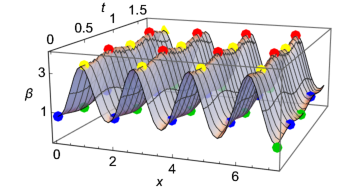}
\includegraphics[width=0.49\textwidth]{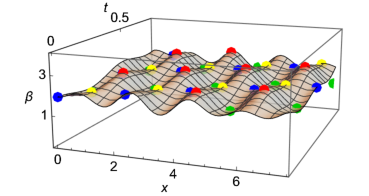} \\
5c. $h_1=100$, $h_0=81$.   $\phantom{mmmmmmmmmm}$  5d. $h_1=100$, $h_0=156$. \\
\includegraphics[width=0.49\textwidth]{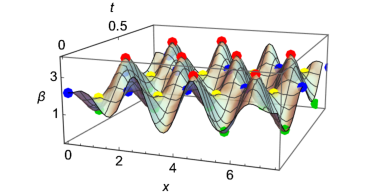}
\includegraphics[width=0.49\textwidth]{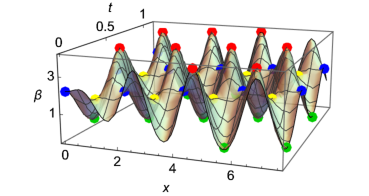}\\
5e. $h_1=210$, $h_0=469$.   $\phantom{mmmmmmmmmmmmmmmmmmmmm}$ \\
\includegraphics[width=0.4\textwidth]{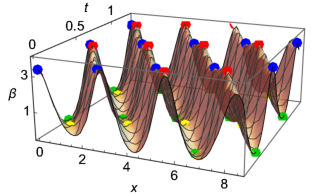}
$\qquad\quad$
\includegraphics[width=0.4\textwidth]{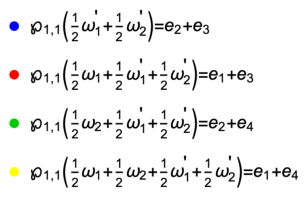} 
\end{figure}

%------------------------------------------------------------------
\subsection{Genus $3$}
The hamiltonian system of the KdV equation in $\M_3^\circ$ possesses
 the spectral curve 
 \begin{equation}\label{G3HamCanon} 
-w^2 + \rmb^2 z^7 + c_5 z^5  + c_4 z^4 + c_3 z^3 + h_2 z^2 + h_1 z + h_0 = 0,
 \end{equation}
 with all real parameters, chosen in such a way that all seven  branch points are real.
According to \eqref{JIPb}, uniformization of \eqref{G3HamCanon} is given by
\begin{gather}\label{UniformG3}
\begin{split}
&z^3 = z^2 \wp_{1,1}(u) + z \wp_{1,3}(u) + \wp_{1,5}(u), \\
&w = - \tfrac{1}{2} \rmb \big(z^2 \wp_{1,1,1}(u) + z \wp_{1,1,3}(u) + \wp_{1,1,5}(u)\big),
\end{split}
\end{gather}
where $u = (u_1,u_3,u_5)$, and all components are complex, in general.

Bounded real-valued solutions in the case of $N=3$ are
\begin{align}\label{betaKdVN3}
&\beta (\x,\rmt) = -\rmb \wp_{1,1} (- \rmb \x + C_1, - \rmb \rmt + C_3, C_5),\\
\intertext{and by the reality conditions,}
 \notag
\begin{split}
&\qquad C_1 = \tfrac{1}{2}\omega'_{1,1} + \tfrac{1}{2}\omega'_{1,2} + \tfrac{1}{2}\omega'_{1,3}, \\
&\qquad C_3 = \tfrac{1}{2}\omega'_{3,1} + \tfrac{1}{2}\omega'_{3,2} + \tfrac{1}{2}\omega'_{3,3}, \\
&\qquad C_5 = -\mathrm{c}_5 + \tfrac{1}{2}\omega'_{5,1} + \tfrac{1}{2}\omega'_{5,2} + \tfrac{1}{2}\omega'_{5,3},
\quad \mathrm{c} \in \Real,
\end{split}
\end{align}
where $(\omega'_{1,k},\omega'_{3,k},\omega'_{5,k})^t = \omega'_k$. Recall that
$\omega'_1$, $\omega'_2$, $\omega'_3$  are three columns of the period matrix $\omega'$.
On $\tfrac{1}{2}  \omega'_1 +  \tfrac{1}{2}  \omega'_2 +  \tfrac{1}{2}  \omega'_3 + \ReN \bar{\Jac}$ 
 the function $\wp_{1,1}$ is real-valued and bounded, with the critical values:
\begin{gather}\label{WP11CritValN3}
\begin{split}
&\wp_{1,1}(\tfrac{1}{2} \omega'_1 + \tfrac{1}{2} \omega'_2 + \tfrac{1}{2} \omega'_3) = e_1 + e_4  + e_5, \\
&\wp_{1,1}(\tfrac{1}{2} \omega_1 + \tfrac{1}{2} \omega'_1 + \tfrac{1}{2} \omega'_2 + \tfrac{1}{2} \omega'_3) 
= e_2 + e_4 + e_5,\\
&\wp_{1,1}(\tfrac{1}{2} \omega_2 + \tfrac{1}{2} \omega'_1 + \tfrac{1}{2} \omega'_2 + \tfrac{1}{2} \omega'_3) 
= e_1 + e_3 + e_5,\\
&\wp_{1,1}(\tfrac{1}{2} \omega_1 + \tfrac{1}{2} \omega_2 
+ \tfrac{1}{2} \omega'_1 + \tfrac{1}{2} \omega'_2 + \tfrac{1}{2} \omega'_3) = e_2 + e_3 + e_5,\\
&\wp_{1,1}(\tfrac{1}{2} \omega_3 
+ \tfrac{1}{2} \omega'_1 + \tfrac{1}{2} \omega'_2 + \tfrac{1}{2} \omega'_3) = e_1 + e_4 + e_6,\\
&\wp_{1,1}(\tfrac{1}{2} \omega_1 + \tfrac{1}{2} \omega_3 
+ \tfrac{1}{2} \omega'_1 + \tfrac{1}{2} \omega'_2 + \tfrac{1}{2} \omega'_3) = e_2 + e_4 + e_6,\\
&\wp_{1,1}(\tfrac{1}{2} \omega_2 + \tfrac{1}{2} \omega_3 
+ \tfrac{1}{2} \omega'_1 + \tfrac{1}{2} \omega'_2 + \tfrac{1}{2} \omega'_3) = e_1 + e_3 + e_6,\\
&\wp_{1,1}(\tfrac{1}{2} \omega_1 + \tfrac{1}{2} \omega_2 + \tfrac{1}{2} \omega_3 
+ \tfrac{1}{2} \omega'_1 + \tfrac{1}{2} \omega'_2 + \tfrac{1}{2} \omega'_3) = e_2 + e_3 + e_6.
\end{split}
\end{gather}
However, projection to the subspace $- (\rmb \x,\rmb \rmt,\mathrm{c})^t +  \tfrac{1}{2}\omega'_{1} 
+  \tfrac{1}{2}\omega'_{2} +  \tfrac{1}{2}\omega'_{3}$, 
which serves as the domain of the quasi-periodic KdV solution $\beta$,
could contain not more than one of the values \eqref{WP11CritValN3}, if $\mathrm{c}$
is one of half-periods constructed from $\tfrac{1}{2}\omega_{1}$, $\tfrac{1}{2}\omega_{2}$, 
$\tfrac{1}{2}\omega_{3}$. 

Let $\M_3^\circ$ be fixed by $4\rmb^2 =1$,  $c_6=0$, $c_5 = -84$,  $c_4=-160$, $c_3 = 7250$.
Let the hamiltonians $h_2$, $h_1$, $h_0$ attain six values with different mutual positions of branch points.
Computed by \eqref{betaKdVN3} with $\rmb = 1/2$ and $\mathrm{c}_5=0$,
shapes of  $\beta(\x,0)$, $\x \in [0,-20 (\omega_{1,1}+\omega_{1,3})]$,
and $\beta(0,\rmt)$, $\rmt\in [0,10]$ or $\rmt\in [0,20]$,  are presented below. \\
\begin{tabular}{l}
$\vphantom{\displaystyle A^{\dfrac{A}{A}}}h_2 = 8178$, $h_1 = -202052$, $h_0 = 111126;\qquad$ 
$e_1 \approx -10.41$,  $e_2 \approx -10.30$, \\ $e_3 \approx -9.77$  $e_4 \approx 0.57$,  
 $e_5 \approx 6.70$,  $e_6 \approx 6.75$,  $e_7 \approx 16.47$; \\
\parbox[]{0.9\paperwidth}{ 
$\quad$ \includegraphics[width=0.84\textwidth]{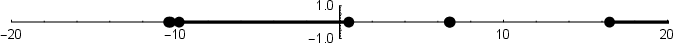} } \\
\includegraphics[width=0.37\textwidth]{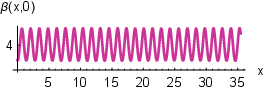}
\includegraphics[width=0.62\textwidth]{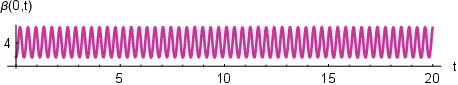} 
\end{tabular}\\
\begin{tabular}{l}
$h_2 = 16678$, $h_1 = -148832$, $h_0 = -338529;\qquad$ 
$e_1 \approx -11.44$,  $e_2 \approx -11.42$, \\ $e_3 \approx -5.67$  $e_4 \approx -2.26$,  
 $e_5 \approx 5.41$,  $e_6 \approx 9.33$,  $e_7 \approx 16.04$; \\
\parbox[]{0.9\paperwidth}{ 
$\quad$ \includegraphics[width=0.84\textwidth]{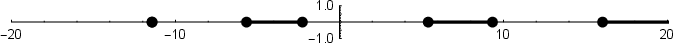} } \\
\includegraphics[width=0.37\textwidth]{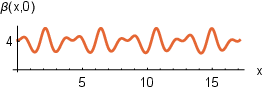}
\includegraphics[width=0.62\textwidth]{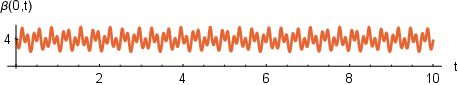} 
\end{tabular}\\
\begin{tabular}{l}
$h_2 = 12 678$, $h_1 = -172 935$, $h_0 = -12782;\qquad$  
$e_1 \approx -11.95$,  $e_2 \approx -9.07$, \\ $e_3 \approx -8.97$  $e_4 \approx -0.07$,  
 $e_5 \approx 4.98$,  $e_6 \approx 8.84$,  $e_7 \approx 16.23$;\\
\parbox[]{0.9\paperwidth}{
$\quad$ \includegraphics[width=0.84\textwidth]{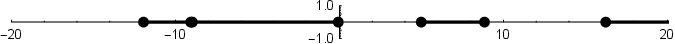} } \\
\includegraphics[width=0.37\textwidth]{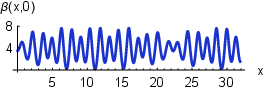}
\includegraphics[width=0.62\textwidth]{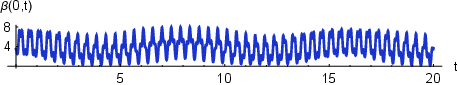} 
\end{tabular}\\
\begin{tabular}{l}
$h_2 = 21460$, $h_1 = -120322$, $h_0 = -287405;\qquad$
$e_1 \approx -12.88$,  $e_2 \approx -9.31$, \\ $e_3 \approx -6.21$  $e_4 \approx -2.16$,  
 $e_5 \approx 4.31$,  $e_6 \approx 10.58$,  $e_7 \approx 15.67$;\\
\parbox[]{0.9\paperwidth}{
$\quad$  \includegraphics[width=0.84\textwidth]{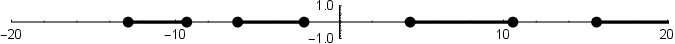} } \\
\includegraphics[width=0.37\textwidth]{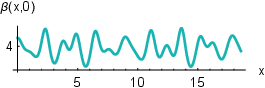}
\includegraphics[width=0.62\textwidth]{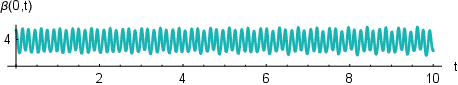} \\
%\end{tabular}\\
%\begin{tabular}{l}
$h_2 = 22460$, $h_1 = -81268$, $h_0 = -267380\qquad$ 
$e_1 \approx -12.21$,  $e_2 \approx -10.94$, \\
 $e_3 \approx -3.57$  $e_4 \approx -3.56$,  $e_5 \approx 3.62$, $e_6 \approx 11.23$,  $e_7 \approx 15.44$;\\
 \parbox[]{0.9\paperwidth}{
$\quad$  \includegraphics[width=0.84\textwidth]{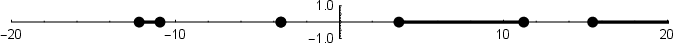} } \\
\includegraphics[width=0.37\textwidth]{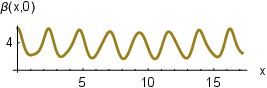}
\includegraphics[width=0.62\textwidth]{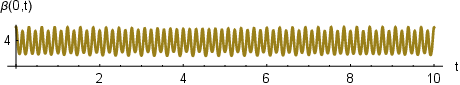} 
\end{tabular}\\
\begin{tabular}{l}
$h_2 = 32678$, $h_1 = -63757$, $h_0 = -372539\qquad$ 
$e_1 \approx -13.85$,  $e_2 \approx -6.64$, \\
 $e_3 \approx -5.10$  $e_4 \approx -5.08$,  $e_5 \approx 3.35$, $e_6 \approx 13.60$,  $e_7 \approx 13.72$.\\
\parbox[]{0.9\paperwidth}{
$\quad$  \includegraphics[width=0.84\textwidth]{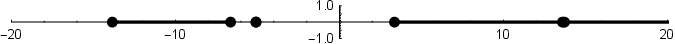} } \\
\includegraphics[width=0.37\textwidth]{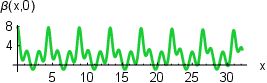}
\includegraphics[width=0.62\textwidth]{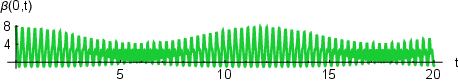} 
\end{tabular}

\medskip
On fig.~\ref{CritValG3} shapes of 
$\beta(\x,\rmt)$, $\x \in [0,6]$, $\rmt\in [0,0.6]$ at some chosen values of $h_2$, $h_1$, $h_0$ are presented. 
\begin{figure}[h]
\caption{$N=3$. Quasi-periodic waves $\beta(\x,\rmt)$,  $\x \in [0,6]$, $\rmt \in [0,0.6]$.} \label{CritValG3}
\parbox[]{0.33\paperwidth}{6a. $h_2 = 8178 \vphantom{\displaystyle a^{\dfrac{a}{a}}}$, $h_1 = -202052$, \\ 
$\phantom{mm}  h_0 = 111126$.}\parbox[]{0.33\paperwidth}{6b. $h_2 = 16678 \vphantom{\displaystyle a^{\dfrac{a}{a}}}$,
$h_1 = -148832$, \\ $\phantom{mm} h_0 = -338529$.} \\
\includegraphics[width=0.49\textwidth]{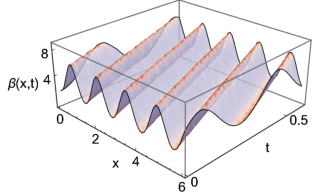}\ \ 
\includegraphics[width=0.49\textwidth]{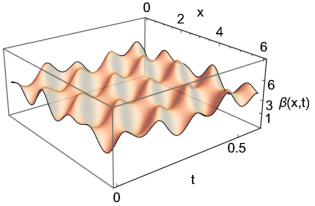} 
%\end{figure}
%\begin{figure}[h]
\parbox[]{0.33\paperwidth}{6c. $h_2 = 12 678$, $h_1 = -172 935$, \\ 
$\phantom{mm} h_0 = -12782$.}\parbox[]{0.33\paperwidth}{6c. $h_2 = 21460$, $h_1 = -120322$ \\ 
$\phantom{mm} h_0 = -287405$.} \\
\includegraphics[width=0.49\textwidth]{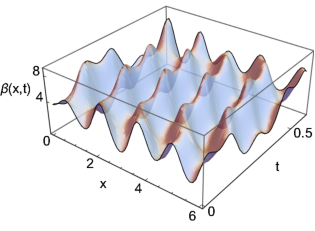}\ \ 
\includegraphics[width=0.49\textwidth]{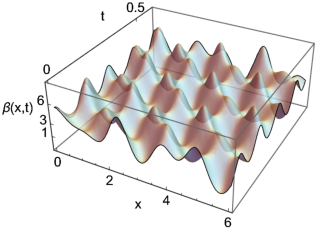}
%\end{figure}
%\begin{figure}[h]
\parbox[]{0.33\paperwidth}{6e. $h_2 = 22460$, $h_1 = -81268$, \\ 
$\phantom{mm} h_0 = -267380$.}\parbox[]{0.33\paperwidth}{6f. $h_2 = 32678$, 
$h_1 = -63757$, \\ $\phantom{mm} h_0 = -372539$.} \\
\includegraphics[width=0.49\textwidth]{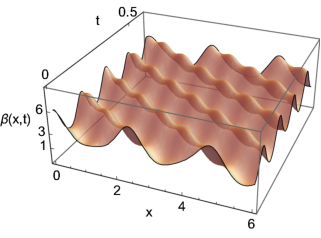}\ \ 
\includegraphics[width=0.49\textwidth]{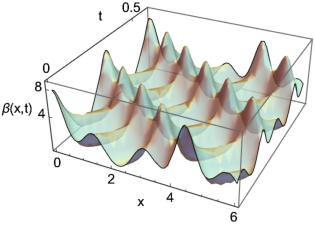}
\end{figure}

%=============================

\end{document}